\documentclass[10pt]{article} %
\usepackage[preprint]{tmlr}

\usepackage{hyperref}
\usepackage{url}

\usepackage{amsmath}
\usepackage{amssymb}
\usepackage{amsthm}
\usepackage{mathtools}
\DeclareMathOperator*{\argmin}{argmin}

\let\originalleft\left
\let\originalright\right
\renewcommand{\left}{\mathopen{}\mathclose\bgroup\originalleft}
\renewcommand{\right}{\aftergroup\egroup\originalright}

\newtheorem{proposition}{Proposition}

\usepackage{booktabs}
\usepackage{makecell}
\usepackage{array}
\usepackage{multirow}
\usepackage{pifont}
\newcommand{\cmark}{\ding{51}}
\newcommand{\xmark}{\ding{55}}

\usepackage{siunitx}

\usepackage{graphicx}
\usepackage{xcolor}
\usepackage{subcaption}
\usepackage{tikz}
\usetikzlibrary{positioning, calc, shapes.geometric, arrows.meta}
\usetikzlibrary{fadings, shadings} %
\graphicspath{ {./graphics/} }


\usepackage[colors]{cora-macs}

\usepackage{cleveref}
\crefname{figure}{Fig.}{Fig.}
\Crefname{figure}{Fig.}{Fig.}
\crefname{subfigure}{Fig.}{Fig.}
\Crefname{subfigure}{Fig.}{Fig.}
\crefname{section}{Sec.}{Sec.}
\Crefname{section}{Sec.}{Sec.}
\crefname{equation}{}{}
\Crefname{equation}{Eq.}{Eq.}
\crefname{table}{Tab.}{Tab.}
\Crefname{table}{Tab.}{Tab.}
\crefname{appendix}{Appendix}{Appendices}
\Crefname{appendix}{Appendix}{Appendices}

\newcommand{\indexMeasured}{m} %
\newcommand{\indexInitialState}{i} %
\newcommand{\indexDisturbance}{j} %
\newcommand{\indexTime}{k} %

\newcommand{\nominal}{{\text{nom}}}

\newcommand{\setStates}{\mathcal{X}} %
\newcommand{\setOutputs}{\mathcal{Y}} %
\newcommand{\setInputs}{\mathcal{U}} %
\newcommand{\setDisturbances}{\mathcal{W}} %

\newcommand{\vectorStates}{x} %
\newcommand{\vectorOutputs}{y} %
\newcommand{\vectorInputs}{u} %
\newcommand{\vectorInfo}{\zeta} %
\newcommand{\vectorDisturbances}{w} %
\newcommand{\vectorResiduals}{r}
\newcommand{\matrixStatesInputsStacked}{Z} %
\newcommand{\matrixResidualsStacked}{R} %

\newcommand{\controller}{\pi_\nominal} %
\newcommand{\controllerUa}{\pi_*} %

\newcommand{\relCostDif}{\delta_c}

\newcommand{\knowledgeAdv}{\kappa}

\newcommand{\bayesPred}{b^{*}}
\newcommand{\rdc}{\tilde{\rho}} %
\newcommand{\learnability}{\rho} %

\title{To Learn or Not to Learn: A Screening Test for Control Problems Before Reinforcement Learning}

\author{\name Victor Schulte \email victor.schulte@tum.de \\
      \addr Department of Computer Engineering\\
      Technical University of Munich, Garching, Germany
      \AND
      \name Michael Eichelbeck \email michael.eichelbeck@tum.de \\
      \addr Department of Computer Engineering\\
      Technical University of Munich, Garching, Germany
      \AND
      \name Matthias Althoff \email althoff@in.tum.de \\
      \addr Department of Computer Engineering\\
      Technical University of Munich, Garching, Germany}

\def\month{MM}  %
\def\year{YYYY} %
\def\openreview{\url{https://openreview.net/forum?id=XXXX}} %

\begin{document}

\maketitle

\begin{abstract}
Reinforcement learning (RL) can be a powerful alternative to classical control methods when standard model-based control is insufficient, e.g., when deriving a suitable model is intractable or impossible.
However, in many cases,  the choice between model-based and RL-based control is not obvious.
Because training RL agents is computationally expensive, RL should be limited to cases where it is expected to outperform model-based control.
To the best of our knowledge, there exists no approach to automatically screen out control problems before RL training for which RL is not expected to yield any benefit.
In this work, we take the first step towards such an automated screening test.
We present a test that identifies problems for which RL-based control is not expected to be superior in two steps.
Only control problems that pass both steps should be considered for RL.
Part~I evaluates the suitability of a provided model for model-based control by analyzing the impact of its uncertainties on the control problem.
For this, we use reachset-conformant system identification combined with simulation-based analysis.
If the uncertainties are found to be relevant in Part~I, Part~II then evaluates their learnability using a statistically calibrated dependence analysis.
We apply our screening test to several benchmarks, reliably identifying control problems for which RL-based control does not measurably improve upon model-based control, thereby achieving a speed-up of up to $21\times$ compared to trying RL.
\end{abstract}

\section{Introduction}

In control engineering, controllers are typically designed based on a system model~\citep{ogata_ModernControlEngineering_2010,franklin_FeedbackControlDynamic_2019}.
For certain problems, synthesizing a suitable model can be intractable, rendering standard controller design impossible.
In this case, reinforcement learning (RL) offers a powerful alternative that does not require any model a priori~\citep{busoniu_ReinforcementLearningControl_2018,recht_TourReinforcementLearning_2019,hasankhani_ComparisonDeepReinforcement_2021,francois-lavet_IntroductionDeepReinforcement_2018}.
Compared to standard model-based approaches, RL algorithms typically use more complex function approximators, such as deep neural networks, which enable the learning of highly complex control tasks.
This flexibility, however, results in substantial computational effort during training~\citep {tsiamis_LearningControlLinear_2022}.
The additional effort is justified only if no tractable model can be derived or if standard model-based approaches fail to deliver sufficient results.
However, in many cases, an approximate system model can be derived from partial knowledge of the underlying system dynamics.
With such approximate models (e.g., a linearized model of a nonlinear system), model-based control frequently yields sufficient control performance~\citep{qin_SurveyIndustrialModel_2003,cibulka_ModelPredictiveControl_2020,morcego_ReinforcementLearningModel_2023,wang_ComparisonReinforcementLearning_2023}.
Thus, for systems where an approximate model can be identified, the choice between standard model-based and RL-based methods is not obvious and depends on the impact of the model inaccuracy on the control problem.

\begin{figure}[!t]
    \centering
    \begin{tikzpicture}[
    node distance=1cm and 2cm,
    block/.style={draw, rectangle, rounded corners=2pt, align=center},
    testblock/.style={block, minimum height=2.4cm, minimum width=2.7cm, fill=CORAcolorBlue!10!white, align=center},
    arrow/.style={->, >=stealth, thick},
    axisline/.style={->, >=stealth, thick},
    iconlabelfont/.style={font=\normalsize, align=center}
]

\coordinate (icon_top_left) at (-5.575,-0.67);

\node[draw, rectangle, thin, rounded corners=2pt, minimum height=1.75cm, minimum width=2.7cm, inner sep=5pt, fill=black!5!white] (data_frame) at ($(icon_top_left)+(0,-0.33)$) {};

\begin{scope}[shift={($(icon_top_left)+(-0.6,0)$)}, x=0.48cm, y=0.35cm]
  \draw[axisline] (0,-1) -- (0,1);
  \draw[axisline] (0,-1) -- (2.5,-1);
  \draw[thick, olive!90!black!90] plot[domain=0.05:2.1, samples=100]
    (\x, {\x*\x*\x - 2.5*\x*\x + \x + 0.2 + 0.15*sin(deg(9*\x)) + 0.1*sin(deg(17*\x))});
  \draw[thick, lime!80!black] plot[domain=0.05:2.1, samples=100]
    (\x, {\x*\x*\x - 2.5*\x*\x + \x + 0.1 + 0.1*sin(deg(10*\x)) + 0.2*sin(deg(20*\x))});
  \draw[thick, orange!80!black] plot[domain=0.05:2.1, samples=100]
    (\x, {\x*\x*\x - 2.5*\x*\x + \x + 0.0 + 0.1*sin(deg(10*\x)) + 0.2*sin(deg(25*\x))});
  \draw[thick, red!70!black] plot[domain=0.05:2.1, samples=100]
    (\x, {\x*\x*\x - 2.5*\x*\x + \x + 0.2 + 0.05*sin(deg(9*\x)) + 0.1*sin(deg(25*\x))});
  \draw[thick, violet!80!black] plot[domain=0.05:2.1, samples=100]
    (\x, {\x*\x*\x - 2.5*\x*\x + \x + 0.4 + 0.05*sin(deg(9*\x)) + 0.1*sin(deg(15*\x))});
\end{scope}

\node[iconlabelfont] (data-label) at ($(icon_top_left)+(0,-0.75)$) {Measured Data};

\coordinate (icon_center_left) at (-5.575,1.33);

\node[draw, rectangle, thin, rounded corners=2pt, minimum height=1.75cm, minimum width=2.7cm, inner sep=5pt, fill=black!5!white] (model_frame) at ($(icon_center_left)+(0,-0.33)$) {};

\begin{scope}[shift={($(icon_center_left)+(-0.6,0)$)}, x=0.48cm, y=0.35cm]
  \draw[axisline] (0,-1) -- (0,1);
  \draw[axisline] (0,-1) -- (2.5,-1);
  \draw[thick, cyan!90!black!90] plot[domain=0.02:2.1, samples=100]
    (\x, {\x*\x*\x - 2.5*\x*\x + \x + 0.2});
\end{scope}

\node[iconlabelfont] (model-label) at ($(icon_center_left)+(0,-0.75)$) {Model};

\node[draw, rectangle, thin, rounded corners=2pt, minimum height=3.3cm, minimum width=6.2cm, fill=CORAcolorBlue!20!white, inner sep=5pt] (test_frame) at (0.0,0.0) {};

\node[iconlabelfont, anchor=north, inner sep=3pt]
  at ($(test_frame.north)+(0,-0.075)$) {RL Screening Test};

\node[testblock] (impact) at ($(test_frame.center)+(-1.5, -0.25)$) {};
\node[iconlabelfont, anchor=north, inner sep=3pt, align=center]
  at (impact.north) {Knowledge \\ Advantage Test};
\begin{scope}[shift={($(impact.center)+(-0.8,-0.45)$)}, x=0.64cm, y=0.5cm]
  \draw[axisline] (0,-1) -- (0,1);
  \draw[axisline] (0,0) -- (2.5,0);
  \draw[thick, red!75!black] plot[domain=0.02:2.1, samples=100]
    (\x, {0.02*sin(deg(20*\x)) + 0.6/(\x^0.3)*sin(deg(12*\x))});
  \draw[thick, lime!85!black] plot[domain=0.02:2.1, samples=100]
    (\x, {0.02*sin(deg(20*\x)) + 0.14/(\x^0.6)*sin(deg(12*\x))});
\end{scope}

\node[testblock] (learn) at ($(test_frame.center)+(1.5, -0.25)$) {};
\node[iconlabelfont, anchor=north, inner sep=3pt, align=center]
  at (learn.north) {Learnability \\ Test};
\begin{scope}[shift={($(learn.center)+(-1.2,-0.95)$)}, x=0.4cm, y=1.42857cm]
  \draw[axisline] (1,0) -- (1,0.7);
  \draw[axisline] (1,0) -- (5,0);
  \foreach \px/\py in {%
    1.10/0.50, 1.20/0.60, 1.30/0.40, 1.40/0.50, 1.50/0.45, 1.60/0.30, 1.70/0.35, 1.80/0.25, 1.90/0.25,
    2.00/0.20, 2.10/0.25, 2.20/0.15, 2.30/0.23, 2.40/0.15, 2.50/0.125, 2.60/0.06, 2.70/0.125, 2.80/0.05, 2.90/0.15,
    3.00/0.10, 3.10/0.05, 3.20/0.125, 3.30/0.08, 3.40/0.15, 3.50/0.10, 3.60/0.17, 3.70/0.20, 3.80/0.15, 3.90/0.30,
    4.00/0.25, 4.10/0.20, 4.20/0.40, 4.30/0.45, 4.40/0.35, 4.50/0.55, 4.60/0.45, 4.70/0.50, 4.80/0.70, 4.90/0.65}
    \fill[teal!90!black] (\px,\py) circle (0.8pt);
\end{scope}

\coordinate (bend_x) at ($(model_frame.east)!0.5!(test_frame.west)$);
\draw[arrow] (model_frame.east) -| ($(bend_x |- test_frame.west)+(0,0.35)$) -- ($(test_frame.west)+(0,0.35)$);
\draw[arrow] (data_frame.east) -| ($(bend_x |- test_frame.west)+(0,-0.35)$) -- ($(test_frame.west)+(0,-0.35)$);

\node[draw=green!50!black, rectangle, thick, rounded corners=2pt, minimum height=1.2cm, minimum width=2.7cm, fill=green!12!white, text=green!25!black, inner sep=5pt, iconlabelfont, align=center, anchor=west] (rl_benefit_frame) at ($(test_frame.east)+(test_frame.west)-(model_frame.east |- test_frame.west)$) {Possible \\ RL Benefit};
\node[draw=red!60!black, rectangle, thick, rounded corners=2pt, minimum height=1.2cm, minimum width=2.7cm, fill=red!10!white, text=red!25!black, inner sep=5pt, iconlabelfont, align=center, anchor=north] (no_rl_benefit_frame) at ($(test_frame.south)+(0,-0.6)$) {No RL Benefit};

\draw[arrow] (test_frame.east) -- (rl_benefit_frame.west);
\draw[arrow] (test_frame.south) -- (no_rl_benefit_frame.north);

\end{tikzpicture}
    \caption{We propose an automated screening test before using RL in control, requiring only measured data and a dynamic model of any provided control problem.}
    \label{fig:test-overview}
\end{figure}

In previous work, RL is often used for control without comparison to a model-based controller at all~\citep{lillicrap_ContinuousControlDeep_2019,duan_BenchmarkingDeepReinforcement_2016,zhou_DatadrivenMethodFast_2020}, or the choice is made only after implementing and comparing both approaches (see \cref{sec:related_work}).
In both cases, significant computational resources may be wasted on tuning and training an RL agent.
Thus, it would be desirable to have an automated screening test that identifies control problems for which RL-based control is not expected to outperform model-based control, without requiring prior agent training.

We address this by proposing a first step towards such an automated screening test.
The test is based solely on simulation, without any training.
It consists of two parts (\cref{fig:test-overview}) and only requires measured data and a system model.
If the problem passes both parts of the test, it should be considered for RL-based control.
In Part~I, we assess whether the system knowledge provided by the model is sufficient for model-based control by approximating the performance improvement that would result from better system knowledge.
If we find that better system knowledge is beneficial, a learning-based approach, such as RL, can be a promising alternative.
Otherwise, if knowing uncertainties offers little benefit, the current model can be considered sufficient for model-based control.
This may be because the uncertainties are not relevant to the control problem or because the model-based controller can compensate for their effects.

However, with this metric alone, we have no information about the extent to which it is possible to get a better understanding of the system dynamics.
This depends on the type of uncertainty, which can be classified as aleatoric or epistemic.
We distinguish this in Part~II.
Aleatoric uncertainty refers to inherent randomness in the system, while epistemic uncertainty refers to a lack of knowledge about the system~\citep{hullermeier_AleatoricEpistemicUncertainty_2021}.
If the uncertainty is purely aleatoric, trying to learn the uncertainties will yield little information gain.
On the contrary, epistemic uncertainty is learnable, and RL may be able to yield a better controller~\citep{hullermeier_AleatoricEpistemicUncertainty_2021}.
To this end, in Part~II we examine whether the model residuals exhibit a learnable structure.
With this two-part test, we present the first automated screening test to identify systems for which RL is not expected to outperform a model-based controller.
As model predictive control (MPC) is widely regarded as one of the most powerful and flexible model-based control methods~\citep{qin_SurveyIndustrialModel_2003}, we use it as the model-based controller baseline in our test.

\section{Related Work}\label{sec:related_work}

\paragraph{Comparisons of RL and model-based control.}
A large body of work, starting with early comparisons on a power system problem~\citep{ernst_ReinforcementLearningModel_2009}, compares RL-based control with model-based control, most often MPC, for specific applications.
Examples include building energy systems~\citep{wang_ComparisonReinforcementLearning_2023}, greenhouse climate control~\citep{morcego_ReinforcementLearningModel_2023}, adaptive cruise control~\citep{lin_ComparisonDeepReinforcement_2021}, ocean current turbines~\citep{hasankhani_ComparisonDeepReinforcement_2021}, energy storage management~\citep{ginzburg-ganz_ComparativeAnalysisOptimal_2025}, nonlinear continuous processes~\citep{rajpoot_ComparisonReinforcementLearning_2024}, and legged locomotion~\citep{akki_BenchmarkingModelPredictive_2025}.
Some of these studies relate their outcome to the model error.
\citet{lin_ComparisonDeepReinforcement_2021} observe that RL outperforms MPC for large modeling errors and performs similarly for small ones.
\citet{koryakovskiy_BenchmarkingModelfreeModelbased_2017} find for a cart-pendulum that nonlinear MPC is superior if the uncertainties can be removed by identifying the system parameters. 
Otherwise, RL outperforms nonlinear MPC with an inaccurate model beyond a break-even point.
These findings motivate our knowledge advantage test, which quantifies the impact of model uncertainty on the control performance without training an RL agent.
Benchmark suites such as safe-control-gym~\citep{yuan_SafeControlGymUnifiedBenchmark_2022} and PC-gym~\citep{bloor_PCGymBenchmarkEnvironments_2026} facilitate such comparisons.
However, these studies comparing RL to model-based control reach their conclusion only after an RL agent has been tuned and trained.
This is exactly the effort that our test aims to avoid, as training RL agents is computationally expensive and energy-intensive~\citep {lacoste_QuantifyingCarbonEmissions_2019}.

\paragraph{Algorithm selection.}
Choosing the most suitable algorithm for a problem instance before running it is known as algorithm selection~\citep{rice_AlgorithmSelectionProblem_1976,kerschke_AutomatedAlgorithmSelection_2019}.
Algorithm selectors are typically learned from the performance of the candidate algorithms on a large set of previously solved instances.
Automated RL optimizes design choices within RL, such as the algorithm and its hyperparameters, but presupposes that RL is used~\citep{parker-holder_AutomatedReinforcementLearning_2022}.
Similarly, hardness measures quantify how difficult a Markov decision process is for RL~\citep{conserva_HardnessMarkovDecision_2022}, but not whether RL outperforms a model-based controller.
In contrast, our test assesses a specific control problem and does not require a set of previously solved problems.

\paragraph{Combining MPC and learning.}
Several works combine model-based control and RL instead of choosing between them~\citep{gorges_RelationsModelPredictive_2017,reiter_SynthesisModelPredictive_2025}.
For instance, a parameterized MPC can serve as the policy of an RL agent~\citep{gros_DataDrivenEconomicNMPC_2020}, a differentiable MPC can be learned end-to-end~\citep{amos_DifferentiableMPCEndtoend_2018} or embedded in an actor-critic architecture~\citep{romero_ActorCriticModelPredictive_2025}, and planning can be combined with a learned model and value function~\citep{hansen_TemporalDifferenceLearning_2022}.
Residual RL learns a correction on top of a provided controller~\citep{silver_ResidualPolicyLearning_2018,johannink_ResidualReinforcementLearning_2019,nahrendra_RetroRLReinforcingNominal_2022a,ankile_ImitationRefinementResidual_2025}, e.g., to compensate for the model-plant mismatch of a nonlinear MPC~\citep{koryakovskiy_ModelPlantMismatchCompensation_2018}.
Learning-based MPC instead learns a model of the residual dynamics from data, e.g., using Gaussian processes, and uses it for prediction~\citep{hewing_CautiousModelPredictive_2020,hewing_LearningBasedModelPredictive_2020}.
These approaches presuppose that a learning component improves upon the nominal controller.
Our test is complementary, as it can be used to determine in advance whether such a learning component is worth its training cost.

\paragraph{Control performance under uncertainty and model validation.}
Work on identification for control~\citep{gevers_IdentificationControlEarly_2005,hjalmarsson_ExperimentDesignClosedloop_2005} and on the objective mismatch in model-based RL~\citep{lambert_ObjectiveMismatchModelbased_2020} emphasizes that model errors matter through their effect on the control performance.
The effect of uncertainty on MPC has been analyzed theoretically, e.g., in terms of stability and performance for uncertain linear systems~\citep{liu_StabilityPerformanceAnalysis_2024} or in terms of regret under imperfect inputs~\citep{liu_RegretModelPredictive_2025}.
In online control, regret is often measured with respect to an oracle controller that knows all future disturbances~\citep{goel_RegretOptimalEstimationControl_2023}.
\Citet{yu_PowerPredictionsOnline_2020} show, for linear-quadratic control, that the dynamic regret of MPC decreases with the number of predicted disturbances.
For building climate control, \citet{oldewurtel_UseModelPredictive_2012} use an MPC with perfect knowledge of future disturbances as a performance bound.
In stochastic programming, the benefit of such perfect knowledge is known as the expected value of perfect information, which is measured relative to the optimal decision under uncertainty~\citep[Ch.~4]{birge_IntroductionStochasticProgramming_2011}.
In process control, the performance of an installed controller is usually assessed against a benchmark, such as minimum-variance control estimated from routine operating data~\citep{harris_AssessmentControlLoop_1989} (see~\citet{qin_ControlPerformanceMonitoring_1998,jelali_OverviewControlPerformance_2006} for surveys).
Our approach builds on these findings and is similar in spirit to randomized algorithms for robustness analysis~\citep{calafiore_SurveyRandomizedAlgorithms_2007,tempo_RandomizedAlgorithmsAnalysis_2013}.
Specific types of uncertainty can be compensated by design choices, e.g., bounded disturbances by robust control~\citep{zhou_RobustOptimalControl_1996} or tube-based MPC~\citep{mayne_RobustModelPredictive_2005}, and stochastic disturbances by stochastic MPC~\citep{mesbah_StochasticModelPredictive_2016}.
However, if uncertainties have a learnable structure, learning-based control may outperform a more conservative robust approach.
This motivates Part~II, our learnability test.
In system identification, a model is validated by testing whether its residuals are uncorrelated with past inputs~\citep{ljung_SystemIdentificationTheory_1999,billings_CorrelationBasedModel_1986}.
Model-plant mismatch in MPC can be detected from closed-loop data, e.g., from partial correlations between model residuals and manipulated inputs~\citep{badwe_DetectionModelplantMismatch_2009}, and distinguished from disturbance model mismatch~\citep{harrison_DiscriminatingDisturbanceProcess_2009}.
Unlike our test, these methods do not quantify the impact of the model uncertainty on the control problem.

\paragraph{Aleatoric and epistemic uncertainty.}
The distinction between aleatoric and epistemic uncertainty is well established~\citep{derkiureghian_AleatoryEpistemicDoes_2009,hullermeier_AleatoricEpistemicUncertainty_2021}.
In machine learning, both types are typically quantified from the predictive distribution of a trained model~\citep{senge_ReliableClassificationLearning_2014,chan_EstimatingEpistemicAleatoric_2024,wimmer_QuantifyingAleatoricEpistemic_2023}, which may come from a Bayesian neural network~\citep{kendall_WhatUncertaintiesWe_2017,depeweg_DecompositionUncertaintyBayesian_2018} or, in model-based RL, from an ensemble of probabilistic dynamics models~\citep{chua_DeepReinforcementLearning_2018}.
Similar to our decomposition, \citet{lahlou_DEUPDirectEpistemic_2023} define epistemic uncertainty as the excess risk of a predictor relative to the optimal point predictor.
Our learnability test does not require training a predictor.

\paragraph{Dependence measures.}
Statistical dependence can be quantified by, e.g., distance correlation~\citep{szekely_MeasuringTestingDependence_2007,huang_StatisticallyNumericallyEfficient_2022}, the Hilbert--Schmidt independence criterion~\citep{gretton_MeasuringStatisticalDependence_2005,gretton_KernelStatisticalTest_2007}, estimators of mutual information~\citep{kraskov_EstimatingMutualInformation_2004,gao_EfficientEstimationMutual_2015,belghazi_MutualInformationNeural_2018,poole_VariationalBoundsMutual_2019,purkayastha_FastMIFastConsistent_2024}, or the randomized dependence coefficient (RDC)~\citep{lopez-paz_RandomizedDependenceCoefficient_2013}, which is based on the Hirschfeld--Gebelein--R\'enyi maximal correlation~\citep{renyi_MeasuresDependence_1959} (see~\citet{desiqueirasantos_ComparativeStudyStatistical_2014} for a comparison).
We use the RDC as a measure of dependence to analyze the structure of model uncertainties.

In summary, existing comparisons decide between RL and model-based control only after training RL, and existing residual tests diagnose a model rather than the expected benefit of learning. 
No method screens control problems before RL training.

\section{Preliminaries}\label{sec:preliminaries}

\paragraph{Reinforcement Learning}
RL is a framework for sequential decision-making based on a Markov decision process.
At each time step, an RL agent receives state feedback $s$ from its environment and takes an action $a$ according to a policy $\pi$.
Whereas states are typically denoted by $s$ and actions by $a$ in the RL community, for the remainder of the paper, we use the control-theoretic notation, in which states are denoted by $x$ and actions by $u$.
Through interactions with the environment, the agent optimizes its policy to maximize the expected reward.
RL does not require any prior knowledge of the system, and, as the policy can be learned with deep neural networks, RL can encode highly complex policies~\citep{francois-lavet_IntroductionDeepReinforcement_2018,recht_TourReinforcementLearning_2019}.

\paragraph{Reachset-Conformant System Identification}
In this work, we employ a system identification method that generates reachset-conformant models~\citep{roehm_ModelConformanceCyberPhysical_2019}.
We assume that a model of a dynamic system is provided or can be identified based on data~\citep{ljung_SystemIdentificationTheory_1999}, with
\begin{equation}\label{eq:model}
    \begin{split}
            \vectorStates(\indexTime+1) = f(\vectorStates(\indexTime), \vectorInputs(\indexTime), \vectorDisturbances(\indexTime)) , \\
            \vectorOutputs(\indexTime) = g(\vectorStates(\indexTime), \vectorInputs(\indexTime), \vectorDisturbances(\indexTime)) ,
    \end{split}
\end{equation}
where $\vectorStates(\indexTime)$ is the state vector at time step $\indexTime$, and $\vectorOutputs(\indexTime)$ the output vector.
The input vector is denoted by $\vectorInputs(\indexTime)$ and the disturbance vector by $\vectorDisturbances(\indexTime)$.
We require the reachable set $\setOutputs(\indexTime)$ at time step $\indexTime$.
For that, we introduce $\phi$, which describes the output of \cref{eq:model} at time step $\indexTime$ when starting at the initial state $x(0)$, subject to the input trajectory $\vectorInputs(\cdot)$ and disturbance trajectory $\vectorDisturbances(\cdot)$.
For all times, $\vectorInputs(\indexTime) \in \setInputs$ and $\vectorDisturbances(\indexTime) \in \setDisturbances$.
It is defined as
\begin{equation}\label{eq:reachset}
    \setOutputs(\indexTime) = \{ \phi (\indexTime;x(0),\vectorInputs(\cdot),\vectorDisturbances(\cdot)) \ | \ x(0) \in \setStates(0), \ \vectorInputs(\tilde{\indexTime}) \in \setInputs(\tilde{\indexTime}), \ \vectorDisturbances(\tilde{\indexTime}) \in \setDisturbances(\tilde{\indexTime}), \ \tilde{\indexTime} \in [\indexTime]_0\}.
\end{equation}
The sets $\setStates(0)$, $\setInputs$, and $\setDisturbances$ are obtained from reachset conformance~\citep{lutzow_ScalableReachsetConformantIdentification_2024,lutzow_ReachsetConformantSystemIdentification_2025}. 
This method identifies $\setStates(0)$, $\setInputs$, and $\setDisturbances$ such that all measured trajectories lie within the reachable set of the identified system.
We use intervals as a set representation for all sets to enable efficient uniform sampling.

\paragraph{Randomized Dependence Coefficient}
In the context of learnability, we aim to determine how well a random variable $\eta$ can be predicted from another variable $\xi$.
This depends on the degree of statistical dependence between $\xi$ and $\eta$, which can be quantified using various measures.
In this work, we use the randomized dependence coefficient (RDC) due to its ability to detect any nonlinear correlation at low computational cost~\citep {lopez-paz_RandomizedDependenceCoefficient_2013}.
The resulting correlation value lies between $0$ and $1$, and the population value it approximates, the Hirschfeld--Gebelein--R\'enyi maximal correlation~\citep{renyi_MeasuresDependence_1959}, equals zero if and only if the two variables are statistically independent.
The computational complexity of the RDC with respect to the dimension $d$ and sample size $m$ is $\mathcal{O} ( d m \log(m))$, making RDC computationally efficient compared to other standard estimators whose computational complexity is quadratic in $m$~\citep{szekely_MeasuringTestingDependence_2007,gretton_MeasuringStatisticalDependence_2005}.

\section{Screening Test for RL-Based Control}\label{sec:screening_test}

In this section, we provide a detailed description of our two-part screening test before the use of RL in control.

\subsection{Part~I: Knowledge Advantage Test}\label{sec:part1}

The first part of our test analyzes the expected advantage that knowledge of the model uncertainties provides to the controller.
To this end, we synthesize an MPC based on the provided or identified model and another MPC with perfect knowledge and compare their performance.
If perfect knowledge yields little or no benefit, we can conclude that a learning-based controller would offer only a limited advantage.

As perfect knowledge is not available in reality, we emulate perfect knowledge by uniformly sampling disturbance values in a simulation environment using the disturbance set $\setDisturbances$ introduced in \cref{eq:reachset} with $n_\indexDisturbance$ sampled trajectories, and $n_\indexTime$ time steps per trajectory:
\begin{align*}
    \vectorDisturbances^{(\indexDisturbance)}(\indexTime) \sim \text{Uniform} \left( \setDisturbances \right)
    \qquad \forall \ \indexDisturbance \in [n_\indexDisturbance] , \ \indexTime \in [n_\indexTime-1]_0 .
\end{align*}

Using~\cref{eq:model}, we synthesize a nominal MPC $ \vectorInputs(\indexTime)= \controller (\vectorStates(\indexTime)) $, assuming no disturbances ($ \vectorDisturbances(\indexTime)= 0 \ \forall \ \indexTime $).
To evaluate the impact of the uncertainties on $\controller$, we compare it to an oracle MPC $\controllerUa$ that has full knowledge of the sampled trajectory $\vectorDisturbances(\cdot)$.
We ensure analysis across the full state space by sampling initial states uniformly with
\begin{equation*}
     \vectorStates^{(\indexInitialState)}(0) \sim \text{Uniform} \left( \setStates(0) \right) \qquad \forall \ \indexInitialState \in [n_\indexInitialState] ,
\end{equation*}
where $n_\indexInitialState$ is the number of initial states sampled.
The performance of both controllers is compared across all combinations of initial states and disturbance trajectories using the non-negative cost function $ J $ of the controller, which can be chosen freely to fit the provided control objective.
With $ c_\nominal^{(i,j)} $ and $ c_*^{(i,j)} $ being the associated cost of evaluating the nominal and the oracle using $J$, respectively, the relative closed-loop cost difference is
\begin{equation*}
        \relCostDif^{(i,j)} = \begin{cases} 1 - \frac{c_*^{(i,j)}}{c_\nominal^{(i,j)}}, & \text{if } c_\nominal^{(i,j)} \neq 0, \\
    0, & \text{if } c_\nominal^{(i,j)} = 0, \end{cases}
\end{equation*}
resulting in the mean relative cost difference across trajectories
\begin{equation*}
    \knowledgeAdv = \frac{1}{n_{i} \cdot n_{j}} \sum_{i = 1}^{n_{i}} \sum_{j = 1}^{n_{j}} \relCostDif^{(i,j)} ,
\end{equation*}
which we call the knowledge advantage.
Assuming $c_\nominal^{(i,j)} \geq c_*^{(i,j)}$, this definition guarantees $\knowledgeAdv \in [0,1]$.

When $\knowledgeAdv = 0$, the uncertainties have no impact on the control performance and the control problem can be identified as suitable for model-based control.
For the knowledge advantage test, uncertainties are sampled until the knowledge advantage converges.
To this end, we compute $\knowledgeAdv$ batch-wise with batch size $n_b$.
The sampling is terminated once the change in $\knowledgeAdv$ caused by a new batch is below a convergence threshold $\lambda$.
If the knowledge advantage $ \knowledgeAdv $ is low, no further testing is required, and it can be concluded that model-based control is sufficient.
On the contrary, if the knowledge advantage is significant, we investigate how well the uncertainties can be learned by an RL agent in the subsequent learnability test.

\subsection{Part~II: Learnability Test}\label{sec:part2}

If the uncertainties are identified as relevant in the first test, we investigate their learnability.
As mentioned above, uncertainty can be classified into aleatoric and epistemic uncertainty.
We quantify the extent to which additional system knowledge can be learned by identifying the share of epistemic uncertainty.

We use the definition for epistemic and aleatoric uncertainty provided by~\citet[Sec. 2.2-2.3]{hullermeier_AleatoricEpistemicUncertainty_2021}.
They introduce the pointwise Bayes predictor as
\begin{equation} \label{eq:bayes_pred}
    \bayesPred(\vectorInfo(\indexTime-1)) :=
    \argmin_{\tilde{\vectorOutputs} \in \mathbb{R}^{n_\vectorOutputs}}
    \int_{\mathbb{R}^{n_\vectorOutputs}}
    \ell (\vectorOutputs(\indexTime), \tilde{\vectorOutputs})
    \cdot dP(\vectorOutputs(\indexTime) \mid \vectorInfo(\indexTime-1)) ,
\end{equation}
where $\vectorInfo(\indexTime-1) = [\hat{\vectorStates}(\indexTime-1)^\top, \vectorInputs(\indexTime-1)^\top]^\top$ is the information available at time step $\indexTime-1$, $\vectorOutputs(\indexTime) \in \mathbb{R}^{n_\vectorOutputs}$ is the true output, and $\tilde{\vectorOutputs}$ is a candidate output prediction.
Since the state is generally not fully measurable, $\hat{\vectorStates}(\indexTime-1)$ is the state estimate of an observer, computed from $\vectorOutputs(0), \ldots, \vectorOutputs(\indexTime-1)$ and $\vectorInputs(0), \ldots, \vectorInputs(\indexTime-2)$.
In case of full observability $\hat{\vectorStates}(\indexTime-1) = \vectorStates(\indexTime-1)$.
The pointwise Bayes predictor is the theoretically optimal point predictor, which chooses the point prediction $\tilde{\vectorOutputs}$ that minimizes the expected loss $\ell$ for a provided $\vectorInfo(\indexTime-1)$.
In our case, this represents the best possible deterministic system model.
\Cref{eq:bayes_pred} uses Lebesgue integration, as this does not require a density function $p$.
The pointwise Bayes predictor can be used to split uncertainty into aleatoric and epistemic components.
We write the output prediction of our nominal model with $\vectorDisturbances = 0$ as $\hat{\vectorOutputs}(\indexTime)$.
With that, we model uncertainty as the residual between the measured data and the nominal model prediction as
\begin{equation}\label{eq:residual_y}
    \vectorResiduals(\indexTime) = \vectorOutputs(\indexTime) - \hat{\vectorOutputs}(\indexTime) .
\end{equation}
Using~\cref{eq:bayes_pred} and~\cref{eq:residual_y}, we rewrite the residual as
\begin{equation}\label{eq:residual}
    \vectorResiduals(\indexTime) = \underbrace{[\vectorOutputs(\indexTime) - \bayesPred(\vectorInfo(\indexTime-1))]}_{= \epsilon(\indexTime) \ \text{(aleatoric part)}}
    + \underbrace{[\bayesPred(\vectorInfo(\indexTime-1)) - \hat{\vectorOutputs}(\indexTime)]}_{ = B(\indexTime) \ \text{(epistemic part)} } .
\end{equation}
Therein, the first term describes the aleatoric part as the difference between the optimal pointwise prediction and the true outcome.
This stems from the fact that~\cref{eq:bayes_pred} is a pointwise predictor, whereas the system is modeled as a stochastic process that outputs a distribution.
The second term describes the learnable, epistemic part as the difference between the current predictor and the optimal pointwise predictor.
In the following, all equalities involving conditional expectations hold almost surely.
\begin{proposition}[Learnable model error and residual dependence]\label{prop:dependence}
Let the residual be decomposed as in~\cref{eq:residual}, where $\ell$ in~\cref{eq:bayes_pred} is the squared loss, $\mathbb{E}[\|\vectorOutputs(\indexTime)\|^2] < \infty$, and the prediction $\hat{\vectorOutputs}(\indexTime)$ is a measurable function of $\vectorInfo(\indexTime-1)$.
(i) If the epistemic part $B(\indexTime)$ is not constant in $\vectorInfo(\indexTime-1)$, then $\vectorResiduals(\indexTime)$ and $\vectorInfo(\indexTime-1)$ are statistically dependent.
(ii) If the aleatoric part $\epsilon(\indexTime)$ is independent of $\vectorInfo(\indexTime-1)$, the converse holds: If $\vectorResiduals(\indexTime)$ and $\vectorInfo(\indexTime-1)$ are statistically dependent, then $B(\indexTime)$ is not constant in $\vectorInfo(\indexTime-1)$.
\end{proposition}

\begin{proof}
Let $\bar{\vectorOutputs}(\indexTime) = \mathbb{E}[\vectorOutputs(\indexTime) \mid \vectorInfo(\indexTime-1)]$, which exists since $\mathbb{E}[\|\vectorOutputs(\indexTime)\|^2] < \infty$.
We use the squared loss $\ell(\vectorOutputs(\indexTime), \tilde{\vectorOutputs}) = \|\vectorOutputs(\indexTime) - \tilde{\vectorOutputs}\|_2^2$ and add and subtract $\bar{\vectorOutputs}(\indexTime)$ inside the norm before computing the conditional expectation. Since it is a function of $\vectorInfo(\indexTime-1)$, $\bar{\vectorOutputs}(\indexTime) - \tilde{\vectorOutputs}$ can be taken out of the conditional expectation giving
\begin{align*}
    \mathbb{E}\big[\|\vectorOutputs(\indexTime) - \tilde{\vectorOutputs}\|_2^2 \mid \vectorInfo(\indexTime-1)\big]
    = &\|\bar{\vectorOutputs}(\indexTime) - \tilde{\vectorOutputs}\|_2^2
    + 2 (\bar{\vectorOutputs}(\indexTime) - \tilde{\vectorOutputs})^\top \underbrace{\mathbb{E}[\vectorOutputs(\indexTime) - \bar{\vectorOutputs}(\indexTime) \mid \vectorInfo(\indexTime-1)]}_{= 0} \\
    &+ \mathbb{E}\big[\|\vectorOutputs(\indexTime) - \bar{\vectorOutputs}(\indexTime)\|_2^2 \mid \vectorInfo(\indexTime-1)\big] .
\end{align*}
Only the first term depends on $\tilde{\vectorOutputs}$ and it is minimized by $\tilde{\vectorOutputs} = \bar{\vectorOutputs}(\indexTime)$.
Thus, by~\cref{eq:bayes_pred}, the Bayes predictor is the conditional mean $\bayesPred(\vectorInfo(\indexTime-1)) = \bar{\vectorOutputs}(\indexTime)$.
Hence, $\mathbb{E}[\epsilon(\indexTime) \mid \vectorInfo(\indexTime-1)] = 0$, and since $B(\indexTime)$ is a function of $\vectorInfo(\indexTime-1)$,
\begin{equation*}
    \mathbb{E}[\vectorResiduals(\indexTime) \mid \vectorInfo(\indexTime-1)] = \mathbb{E}[\epsilon(\indexTime) \mid \vectorInfo(\indexTime-1)] + \mathbb{E}[B(\indexTime) \mid \vectorInfo(\indexTime-1)] = B(\indexTime) .
\end{equation*}
(i) If $\vectorResiduals(\indexTime)$ and $\vectorInfo(\indexTime-1)$ were independent, then $\mathbb{E}[\vectorResiduals(\indexTime) \mid \vectorInfo(\indexTime-1)] = \mathbb{E}[\vectorResiduals(\indexTime)]$ would be constant, and so would $B(\indexTime)$.
(ii) If $B(\indexTime) = c$ for a constant $c$ and $\epsilon(\indexTime)$ is independent of $\vectorInfo(\indexTime-1)$, then $\vectorResiduals(\indexTime) = c + \epsilon(\indexTime)$ is independent of $\vectorInfo(\indexTime-1)$.
\end{proof}

Hence, if $\epsilon(\indexTime)$ is independent of $\vectorInfo(\indexTime-1)$, residual dependence indicates a model error $B(\indexTime)$ that is learnable from $\hat{\vectorStates}(\indexTime-1)$ and $ \vectorInputs(\indexTime-1)$.
Given this relation, we use residual independence to identify systems in which RL-based control is expected to provide no benefit.
If the distribution of $\epsilon(\indexTime)$ varies with $\vectorInfo(\indexTime-1)$, e.g., for heteroscedastic noise, dependence can occur even for constant $B(\indexTime)$.
This does not change the conclusion of our test, as an RL agent may learn to avoid regions of high uncertainty.
Other types of uncertainty, such as independent and identically distributed (i.i.d.) noise or a constant offset, can be compensated by model-based approaches such as Kalman filters~\citep{kalman_NewApproachLinear_1960}, offset-free MPC~\citep{pannocchia_DisturbanceModelsOffsetfree_2003}, or robust control~\citep{zhou_RobustOptimalControl_1996} and do not justify the use of RL-based control.

For the learnability test, we compute the residual for every time step $\indexTime \in [n_\indexTime-1] $ of each measured trajectory $\indexMeasured \in [n_\indexMeasured]$.
For the dependence analysis, we use the randomized dependence coefficient (RDC) for the reasons outlined in~\cref{sec:preliminaries}.
It yields the dependence coefficient $\rdc = \operatorname{RDC}(\matrixResidualsStacked, \matrixStatesInputsStacked) \in [0,1]$, where the matrix $\matrixResidualsStacked$ contains all residuals $\vectorResiduals(\indexTime)$ and $\matrixStatesInputsStacked$ contains all vectors $\vectorInfo(\indexTime-1)$.

Since the empirical RDC may be positive even for independent variables, we calibrate $\rdc$ with a permutation test of the null hypothesis that $\vectorResiduals(\indexTime)$ and $\vectorInfo(\indexTime-1)$ are independent.
Shuffling individual samples is not suitable, since the residuals within a trajectory may be autocorrelated and both the residuals and the state estimates may change systematically over the course of a trajectory.
Instead, we re-pair entire trajectories while keeping the time steps aligned.
For a random permutation $\sigma$ of $\{1, \ldots, n_\indexMeasured\}$, the residual sequence $\vectorResiduals^{(\sigma(\indexMeasured))}(1), \ldots, \vectorResiduals^{(\sigma(\indexMeasured))}(n_\indexTime-1)$ of trajectory $\sigma(\indexMeasured)$ is paired with the sequence $\vectorInfo^{(\indexMeasured)}(0), \ldots, \vectorInfo^{(\indexMeasured)}(n_\indexTime-2)$ of trajectory $\indexMeasured$ for all $\indexMeasured$, and the RDC $\rdc_b$ is computed.
This preserves the temporal correlation within each sequence and any time dependence shared by all trajectories, while removing the dependence between $\vectorResiduals(\indexTime)$ and $\vectorInfo(\indexTime-1)$ of the same trajectory.
Repeating this for $b = 1, \ldots, n_\text{perm}$ yields RDC values $\rdc_1, \ldots, \rdc_{n_\text{perm}}$ attainable under independence.
For a significance level $\alpha$, let the critical value $q_\alpha$ be the $N_\alpha$-th largest of these values, where $N_\alpha = \lfloor \alpha (n_\text{perm} + 1) \rfloor$, which requires $\alpha (n_\text{perm} + 1) \geq 1$.
With this choice, $\rdc > q_\alpha$ holds if and only if the permutation p-value $(1 + |\{ b : \rdc_b \geq \rdc \}|) / (n_\text{perm} + 1)$~\citep{phipson_PermutationPvaluesShould_2010} does not exceed $\alpha$.
We define the learnability as the RDC rescaled relative to this critical value:
\begin{equation}\label{eq:learnability}
    \learnability = \max \left( 0, \frac{\rdc - q_\alpha}{1 - q_\alpha} \right) \in [0,1] ,
\end{equation}
assuming $q_\alpha < 1$.
Hence, $\learnability = 0$ if $\rdc$ does not exceed the values attainable under independence, and $\learnability$ increases linearly with $\rdc$ to $\learnability = 1$ for $\rdc = 1$.
With that, $\learnability$ indicates both whether there is significant dependence and how strong this dependence is.
Then $\alpha (n_\text{perm} + 1)$ is an integer, so that the test rejects with probability exactly $\alpha$ under the null hypothesis~\citep{hemerik_ExactTestingRandom_2018}, provided that the measured trajectories are i.i.d.
Since $\vectorResiduals(\indexTime)$ influences the subsequent states in closed loop, the test is approximate.
If $\learnability > 0$, \Cref{prop:dependence}(ii) implies state-dependent structure in the uncertainties, and a learning-based controller may outperform the nominal model-based controller.

\section{Evaluation}\label{sec:evaluation}

We evaluate our screening test against three control benchmarks from different domains, each with increasing complexity.
Throughout the benchmarks, we introduce several types of uncertainty, namely random noise, parametric uncertainty, and unmodeled dynamics.
After running our test on each benchmark, we evaluate the outcome by training an RL agent and comparing its performance to that of an MPC.
To obtain a good policy, we try Proximal Policy Optimization (PPO)~\citep{schulman_ProximalPolicyOptimization_2017} and Soft Actor-Critic (SAC)~\citep{haarnoja_SoftActorCriticOffPolicy_2018} for each experiment.
Both are widely used for continuous control~\citep{yuan_SafeControlGymUnifiedBenchmark_2022}.
To ensure that the comparison is not decided by a poorly chosen hyperparameter configuration, we perform systematic hyperparameter tuning using the Optuna library~\citep{akiba_OptunaNextgenerationHyperparameter_2019} for the nominal MPC as well as for the RL agents.
All evaluation and benchmark details are reported in Appendices~\ref{app:evaluation_details}--\ref{app:rl}.

\subsection{Main Results}

To show the effectiveness of our screening test across different types of uncertainty, we use the cart pole benchmark~\citep{barto_NeuronlikeAdaptiveElements_1983}, a common control benchmark in which the controller must stabilize a pole balanced on a cart with one-dimensional motion (details in~\ref{app:plant_cartpole}).
We evaluate the cart pole benchmark at three levels of uncertainty.
\begin{table}[t]
    \centering
    \caption{Output dimension $n_\vectorOutputs$ of each benchmark, knowledge advantage $\knowledgeAdv$, learnability $\learnability$, and whether our test proposes model-based control as sufficient, together with the closed-loop cost $J$ of MPC, PPO, and SAC (mean $\pm$ standard deviation over five seeds) and the number of seeds of the best RL agent (lowest mean cost) achieving a lower cost than MPC (RL wins).}
    \label{tab:rl_mpc}
    \small
    \setlength{\tabcolsep}{3.5pt}
    \begin{tabular}{lcccccccc}
        \toprule
        & & \multicolumn{3}{c}{\textbf{Test}} & \textbf{MPC} & \textbf{PPO} & \textbf{SAC} & \textbf{RL} \\
        \cmidrule(lr){3-5}\cmidrule(lr){6-6}\cmidrule(lr){7-7}\cmidrule(lr){8-8}%
        \textbf{Benchmark} & $n_\vectorOutputs$ & $\knowledgeAdv$ & $\learnability$ & model-based & $J (\downarrow)$ & $J (\downarrow)$ & $J (\downarrow)$ & \textbf{wins} \\
        \midrule
        Cart pole \textit{(velocity noise)}       & 4 & $0.11$ & --      & \cmark & $\boldsymbol{101.4}$  & $131.6\pm34.1$ & $103.5\pm0.4$                   & 0\,/\,5 \\
        \cmidrule(lr){1-9}
        Cart pole \textit{(state noise)} & 4 & $0.88$ & $0.00$  & \cmark & $\boldsymbol{88.0}$   & $251.3\pm92.1$ & $93.3\pm2.7$                    & 0\,/\,5 \\
        \cmidrule(lr){1-9}
        Cart pole \textit{(track force)} & 4 & $0.98$ & $0.98$ & \xmark & $104.2$               & diverged       & $\boldsymbol{96.3\pm2.6}$       & \textbf{5\,/\,5} \\
        \cmidrule(lr){1-9}
        CSTR                     & 3 & $0.003$ & --     & \cmark & $\boldsymbol{1332.1}$ & $1345.5\pm5.9$ & $1399.9\pm22.1$                 & 0\,/\,5 \\
        \cmidrule(lr){1-9}
        Quadrotor                & 6 & $0.04$ & --      & \cmark & $\boldsymbol{4.07}$   & $4.59\pm0.34$  & diverged              & 0\,/\,5 \\
        \bottomrule
    \end{tabular}
\end{table}
\begin{table}[t]
    \centering
    \caption{Wall-clock time comparison on identical resources for all benchmarks for which the test predicts no RL advantage. The RL training time includes the total compute time (Optuna tuning plus five-seed training) for the best-performing agent per benchmark. The screening test column contains the time of our two-part test, including reachset-conformant identification.}
    \label{tab:time}
    \small
    \setlength{\tabcolsep}{4pt}
    \begin{tabular}{lccc}
        \toprule
        \textbf{Benchmark} & \textbf{\makecell{Best policy}} & \textbf{\makecell{RL training}} & \textbf{\makecell{Screening test}} \\
        \midrule
        Cart pole \textit{(velocity noise)} & SAC & 3 h 9 min & \textbf{8 min 48 s} \\
        Cart pole \textit{(state noise)} & SAC & 3 h 4 min & \textbf{9 min 27 s} \\
        CSTR & PPO & 2 h 31 min & \textbf{9 min 6 s} \\
        Quadrotor & PPO & 21 min & \textbf{2 min 47 s} \\
        \bottomrule
    \end{tabular}
\end{table}
In the first experiment, \textit{velocity noise}, we add uniformly sampled noise to the velocity $v$ sampled from the interval $ [-0.01, 0.01]$.
Running our test returns a knowledge advantage of $\knowledgeAdv = 0.11$ (see~\cref{tab:rl_mpc}).
This indicates that the uncertainties have little impact on the control problem (see~\cref{fig:low_impact}).
Thus, the MPC can be considered sufficiently performant, and there is little expected performance gain through a learning-based approach.
Training RL agents on this benchmark validates this result: For both PPO and SAC, none of the five trained agents achieves a lower averaged closed-loop cost $J$ than the MPC, as predicted by our test (see~\cref{tab:rl_mpc}). 
As our test is $21\times$ faster to compute than tuning and training the agents (see~\cref{tab:time}), it substantially saves time and computational resources.
In a second experiment, \textit{state noise}, we add disturbances to all four states, sampled uniformly from $[-0.002, 0.002]^4$.
This is expected to increase the impact on the control problem.
In this case, Part~I of our test delivers a knowledge advantage of $\knowledgeAdv = 0.88$, close to the maximum value of $1$, indicating that the MPC is underperforming compared to a controller with perfect knowledge (see~\cref{fig:high_impact}).
Given the significant knowledge advantage, we apply the learnability test to assess whether a learning-based approach can exploit this performance gap.
As the disturbances are randomly sampled, we do not expect to detect a learnable structure.
Indeed, the test returns a learnability of $\learnability = 0.00$.
This suggests that the potential benefit of a learning-based controller is limited for this control problem.
As expected, the \textit{state noise} row of \cref{tab:rl_mpc} shows that RL agents cannot improve on the performance of the nominal MPC.
We observe a substantial $20\times$ reduction in time in our test (see~\cref{tab:time}).
In a third experiment, \textit{track force}, we assume that the true plant features friction between the cart and the surface and an uneven surface.
This results in a track force acting on the cart that the model does not account for.
Our test reveals a high knowledge advantage of $\knowledgeAdv = 0.98$ and a high learnability of $\learnability = 0.98$.
Thus, this system should be considered for learning-based control such as RL.
As seen in the \textit{track force} row of \cref{tab:rl_mpc}, RL-based control can outperform nominal MPC.

\begin{figure}
    \centering
    \begin{subfigure}[t]{0.48\linewidth}
        \centering
        \input{graphics/inv_pend_x_0_low_impact.tex}
        \caption{The knowledge advantage is close to zero, as there is little impact of the disturbances on the nominal controller.}
        \label{fig:low_impact}
    \end{subfigure}%
    \hfill
    \begin{subfigure}[t]{0.48\linewidth}
        \centering
        \input{graphics/inv_pend_x_0_high_impact.tex}
        \caption{The knowledge advantage is significant, as the oracle controller visibly outperforms the nominal controller under disturbance.}
        \label{fig:high_impact}
    \end{subfigure}
    \caption{Cart position of the cart pole benchmark controlled by a nominal and an oracle MPC under sampled disturbance trajectories during the knowledge advantage test. The initial state is the same for all trajectories.}
    \label{fig:cartpole_impact}
\end{figure}

To evaluate our test on a highly nonlinear benchmark, we apply our test algorithm to the control of a continuously stirred tank reactor (CSTR;~\citet{machalek_DynamicEconomicOptimization_2020}) -- a system common in process engineering (details in~\cref{app:plant_cstr}).
The exponential nature of Arrhenius reaction kinetics in a CSTR results in highly nonlinear system dynamics.
In our adaptation, we assume that the reaction rates of the model are inaccurate.
Our knowledge advantage test indicates that this parametric uncertainty yields a knowledge advantage of $\knowledgeAdv = 0.003$.
This indicates no noticeable impact of the model mismatch on the control problem.
Due to the inherently slow dynamics arising from the high inertia of the system relative to the sample time, the MPC has sufficient time to compensate for model inaccuracies.
As shown in \Cref{tab:rl_mpc}, the RL agent is unable to surpass the MPC, supporting the findings of our screening test.

At last, we consider the task of stabilizing a two-dimensional quadrotor~\citep{yuan_SafeControlGymUnifiedBenchmark_2022}, a benchmark with output dimension $n_\vectorOutputs=6$, where the motion of the quadrotor is constrained to the $xz$-plane.
The objective of the control problem is to stabilize the quadrotor above the ground.
Close to the ground, the quadrotor is affected by ground effects, which, in our case, are not accounted for by the model.
The ground effects are modeled as described by~\citet{cheeseman_EffectGroundHelicopter_1955} (details in \ref{app:plant_quad}).
The unmodeled ground effects result in a knowledge advantage of $\knowledgeAdv = 0.04$, indicating that a model-based controller can compensate for their effects.
\Cref{tab:rl_mpc} validates this finding.
The PPO agent closely matches the nominal MPC performance but does not surpass it.
\Cref{tab:time} shows that, with our test, we again achieve a significant time advantage for both the CSTR and the quadrotor compared to training an RL agent.
Across all benchmarks, our test achieves at least a $7$-fold speed-up.

\subsection{Ablation Study}\label{sec:ablation}

\begin{table}[t]
    \centering
    \caption{Wall-clock time of the screening test for the spring-mass-damper benchmark with $n_\vectorOutputs$ outputs. The test is stopped after 2 hours.}
    \label{tab:ablation}
    \small
    \setlength{\tabcolsep}{4pt}
    \begin{tabular}{cccc}
        \toprule
        \multirow{2}{*}{$\boldsymbol{n_\vectorOutputs}$} & \multicolumn{2}{c}{\textbf{Part~I}} & \multirow{2}{*}{\textbf{Part~II}} \\
        \cmidrule(lr){2-3}
        & \textbf{Identification} & \textbf{Knowledge advantage} & \\
        \midrule
        2 & 2 min 18 s & 39 s & 41 s \\
        4 & 13 min 40 s & 54 s & 39 s \\
        6 & 45 min 5 s & 54 s & 42 s \\
        8 & 1 h 12 min 34 s & 1 min 20 s & 46 s \\
        10 & $>$ 2 h & -- & -- \\
        12 & $>$ 2 h & -- & -- \\
        \bottomrule
    \end{tabular}
\end{table}

To evaluate the scalability of our proposed test, we apply it to a spring-mass-damper system that can easily be extended to any number of states.
Each additional mass adds two states and, with them, two outputs to the system: The position and the velocity of the mass.
We assume a model mismatch caused by an underestimation of the spring constant combined with an i.i.d. velocity disturbance (details in~\ref{app:plant_chain}).
\Cref{tab:ablation} summarizes the times for each part of our proposed test.
For a detailed analysis, we split Part~I into reachset-conformant identification and the computation of the knowledge advantage. 
We set a two-hour time limit for the test.
The table shows that reachset-conformant identification is the bottleneck as the number of outputs scales.
Whereas the computation of the knowledge advantage in Part~I and the learnability in Part~II increase only slightly with $n_\vectorOutputs$, the identification alone exceeds the limit of two hours at $n_\vectorOutputs = 10$.
Thus, we conclude that the actual computation of the knowledge advantage and learnability is not the bottleneck for scalability in the test, and any improvements to the reachset-conformant identification directly benefit our test as well.

\section{Limitations}\label{sec:limitations}

\paragraph{Interpretation of a positive test result.}
A positive result of the learnability test indicates that RL may be beneficial, not that it will be.
The extent to which an RL agent or any learning-based controller can exploit the learnable structure of uncertainties depends strongly on the chosen algorithm.
As shown in \Cref{tab:rl_mpc}, depending on the benchmark, either SAC or PPO is the best-performing RL algorithm. 
Even with high learnability, not every RL algorithm can outperform an MPC.
Our test is therefore designed to reliably screen out control problems for which RL is not beneficial, whereas problems that pass the test still require comparison between an RL agent and a model-based controller.
The choice of RL algorithm is not in the scope of our test.

\paragraph{Scalability.}
As shown in \cref{sec:ablation}, the reachset-conformant identification exceeds a two-hour budget for $n_\vectorOutputs \geq 10$, whereas the computation of $\knowledgeAdv$ and $\learnability$ scales well.
As reachset-conformant identification is an active field of research, this bottleneck may be improved by future work.
Alternatively, in future work, a different method for quantifying uncertainties may be used in the test.

\section{Conclusion}\label{sec:conclusion}

In this work, we address the issue of choosing a performant controller for a control problem while avoiding unnecessary complexity. 
Specifically, we have developed the first step towards an automated screening test that identifies control problems for which costly RL training can be avoided. 
The test is based on the premise that, whenever a straightforward model-based approach is sufficient, the high effort and complexity of RL should be avoided.
We show that across all evaluated benchmarks, our two-part test successfully identifies systems that do not benefit from RL.
The test evaluates the impact and learnability of all uncertainties that a provided model does not consider, without requiring RL training.
In our benchmarks, this reduces runtime by up to $21\times$ when RL provides no benefit.
Our test can be universally applied to benchmarks across domains and requires only a model and measured trajectory data.
If no model is available, one can be identified based on the measured data.
The current limitation of the test is the lack of scalability of the reachset-conformant identification method used to quantify uncertainty. 
This may be improved in future work by enhancing the scalability of the method or by choosing a different uncertainty quantification method.

\subsubsection*{Acknowledgements}

This work was funded by the German Research Foundation (DFG, Deutsche Forschungsgemeinschaft) – GRK 3081/1 – Project number 534429653.

\bibliography{references}
\bibliographystyle{tmlr}

\appendix
\raggedbottom
\crefalias{section}{appendix}
\crefalias{subsection}{appendix}

\section{Sensitivity Analysis}\label{app:sensitivity}

\subsection{Knowledge-Advantage Test}\label{app:sensitivity_ka}

The adaptive sampling of the knowledge-advantage test is governed by two hyperparameters, the batch size $n_b$ and the convergence threshold $\lambda$.
To assess how strongly they influence the test outcome, we conduct an ablation study on the cart pole with random disturbances applied to all four states.
We evaluate the grid $n_b \in \{1, 2, 5, 10, 20\}$ and $\lambda \in \{0.005, 0.01, 0.02, 0.05, 0.1, 0.2\}$ with a total budget of $M = 500$ samples.
Exhausting the entire budget yields the value $\knowledgeAdv = 0.86$.
This serves as a reference value because it is computed without batch-wise evaluation or early stopping.
\Cref{fig:ablation} shows the resulting knowledge advantage over the convergence threshold with one line per batch size.
Across the full grid, the estimate remains close to the reference value with a maximum absolute deviation of less than $0.02$.
This shows that the outcome of the test is robust with respect to $n_b$ and $\lambda$, justifying the use of $n_b =5$ and $\lambda=0.05$ throughout all experiments.

\begin{center}
\begin{minipage}{\linewidth}
\centering
\input{graphics/inv_pend_ablation_estimate_vs_threshold.pgf}
\captionof{figure}{Sensitivity of the knowledge advantage $\knowledgeAdv$ with respect to the batch size $n_b$ and the convergence threshold $\lambda$ for the cart pole. The markers denote the mean across five seeds, and the error bars denote the corresponding standard deviation. The dashed line and the shaded band mark the reference obtained by exhausting the sampling budget.}
\label{fig:ablation}
\end{minipage}
\end{center}

\subsection{Learnability Test}\label{app:sensitivity_learnability}

Besides the number of re-pairings $n_\text{perm}$ and the significance level $\alpha$, the learnability test has four hyperparameters.
Three of them belong to the RDC~\citep{lopez-paz_RandomizedDependenceCoefficient_2013}, which computes the largest canonical correlation between $k_\text{RDC}$ random features $\sin(w_l^\top c + b_l)$, $l = 1, \ldots, k_\text{RDC}$, of the copula-transformed samples $c$ of each variable.
To reduce the variance due to the random features, $\rdc$ is the median of $n_\text{RDC}$ RDCs with independently drawn features, and the same features are used for the measured and all re-paired data.
The fourth hyperparameter, the tolerance $\tau$, excludes components that are constant up to floating-point rounding, since the rank-based copula transformation would turn their rounding errors into spurious structure.
Residual component $j$ is excluded if its range over all samples does not exceed $\tau \max_{\indexMeasured, \indexTime} |\vectorStates_j^{(\indexMeasured)}(\indexTime)|$, and component $j$ of $\vectorInfo$ is excluded if its range does not exceed $\tau \max_{\indexMeasured, \indexTime} |\vectorInfo_j^{(\indexMeasured)}(\indexTime)|$.
Since this criterion depends only on the values of each component, it excludes the same components for the measured and all re-paired data.
All experiments in the main text use $k_\text{RDC} = 50$, $s_\text{RDC} = 1/6$, $n_\text{RDC} = 5$, $n_\text{perm} = 999$, $\alpha = 0.05$, and $\tau = \sqrt{\epsilon_\text{mach}} \approx 1.5 \cdot 10^{-8}$.

We assess the sensitivity of $\learnability$ to these hyperparameters on two cart pole data sets with high and low learnability.
The first is the data set of the \textit{track force} scenario, the second contains i.i.d.\ uniform disturbances from $[-0.002, 0.002]^4$ on all four states as in the \textit{state noise} scenario.
We vary one block of hyperparameters at a time and keep the others at their default values, see \cref{tab:app_sens_grid}.
Within a block, all combinations of the listed values are evaluated, except $\alpha = 0.01$ with $n_\text{perm} = 49$, which is undefined since $\alpha (n_\text{perm} + 1) \geq 1$ is required.
Every setting is evaluated with five seeds, which determine the random features and the re-pairings.
For each seed, we additionally run the test with $n_\text{perm} = 4999$ re-pairings as a reference.
Since the re-pairings are drawn sequentially from the seed, the first $n_\text{perm}$ of them form exactly the test with $n_\text{perm}$ re-pairings.

\begin{center}
\begin{minipage}{\linewidth}
\centering
\captionof{table}{Hyperparameter values of the learnability sensitivity analysis. Each block is varied while all other hyperparameters are kept at their default values.}
\label{tab:app_sens_grid}
{\small
\begin{tabular}{llc}
\toprule
\textbf{Hyperparameter} & \textbf{Values} & \textbf{Default} \\
\midrule
\multicolumn{3}{l}{\textit{Permutation test}} \\
Number of re-pairings $n_\text{perm}$ & $\{49, 99, 199, 499, 999, 1999\}$ & $999$ \\
Significance level $\alpha$ & $\{0.01, 0.025, 0.05, 0.1\}$ & $0.05$ \\
\midrule
\multicolumn{3}{l}{\textit{Median}} \\
Number of RDCs $n_\text{RDC}$ & $\{1, 3, 5, 7, 9\}$ & $5$ \\
\midrule
\multicolumn{3}{l}{\textit{RDC features}} \\
Number of features $k_\text{RDC}$ & $\{5, 10, 20, 50, 100\}$ & $50$ \\
Scale $s_\text{RDC}$ & $\{1/24, 1/12, 1/6, 1/3, 2/3, 4/3\}$ & $1/6$ \\
\midrule
\multicolumn{3}{l}{\textit{Tolerance}} \\
Tolerance $\tau$ & \makecell[l]{$\{0, 10^{-15}, 10^{-12}, 1.5 \cdot 10^{-8}, 10^{-6}, 10^{-4}, 10^{-3},$ \\ $\phantom{\{}3 \cdot 10^{-3}, 5 \cdot 10^{-3}, 10^{-2}, 3 \cdot 10^{-2}, 0.1, 0.4, 1\}$} & $1.5 \cdot 10^{-8}$ \\
\bottomrule
\end{tabular}}
\end{minipage}
\end{center}

\paragraph{RDC parameters.}
\Cref{tab:app_sens_rdc} shows that $k_\text{RDC}$ and $s_\text{RDC}$ change the raw RDC considerably.
On the random disturbance data, $\rdc$ grows from $0.022$ for $k_\text{RDC} = 5$ to $0.141$ for $k_\text{RDC} = 100$, although the residuals are independent of $\vectorInfo$.
This is the finite-sample bias of the RDC, which grows with the flexibility of the random features.
The critical value $q_\alpha$ tracks this bias in every cell and stays $0.002$ to $0.021$ above $\rdc$, so that $\learnability = 0$ throughout.
This illustrates the necessity of using the permutation test.
On the track force data, $q_\alpha$ also rises with $k_\text{RDC}$ and $s_\text{RDC}$, from $0.26$ for $k_\text{RDC} = 5$ to $0.47$ for $k_\text{RDC} \geq 50$ and $s_\text{RDC} = 4/3$, as more flexible features fit the re-paired data better.
Nevertheless, $\learnability$ stays between $0.958$ and $0.983$.
For $\learnability > 0$, \cref{eq:learnability} gives $1 - \learnability = (1 - \rdc) / (1 - q_\alpha)$, so for $\rdc \approx 0.99$ the variation of $q_\alpha$ changes $\learnability$ by less than $0.01$.

\begin{center}
\begin{minipage}{\linewidth}
\centering
\captionof{table}{Sensitivity of the learnability test to the RDC parameters $k_\text{RDC}$ and $s_\text{RDC}$. Top: Learnability $\learnability$ on the track force data. Bottom: Raw RDC $\rdc$ and critical value $q_\alpha$ on the random disturbance data, where $\learnability = 0$ in every cell and seed. The default setting is printed in bold.}
\label{tab:app_sens_rdc}
{\footnotesize
\setlength{\tabcolsep}{4pt}
\begin{tabular}{lcccccc}
\toprule
 & \multicolumn{6}{c}{$s_\text{RDC}$} \\
\cmidrule(lr){2-7}
$k_\text{RDC}$ & $1/24$ & $1/12$ & $1/6$ & $1/3$ & $2/3$ & $4/3$ \\
\midrule
\multicolumn{7}{l}{\textit{Track force: $\learnability$}} \\
$5$   & $0.981$ & $0.979$ & $0.975$ & $0.974$ & $0.967$ & $0.958$ \\
$10$  & $0.981$ & $0.982$ & $0.982$ & $0.981$ & $0.979$ & $0.976$ \\
$20$  & $0.982$ & $0.982$ & $0.982$ & $0.981$ & $0.979$ & $0.977$ \\
$50$  & $0.982$ & $0.980$ & $\mathbf{0.979}$ & $0.981$ & $0.980$ & $0.979$ \\
$100$ & $0.983$ & $0.981$ & $0.981$ & $0.981$ & $0.981$ & $0.980$ \\
\midrule
\multicolumn{7}{l}{\textit{Random disturbances: $\rdc$\,/\,$q_\alpha$}} \\
$5$   & $0.024$\,/\,$0.033$ & $0.023$\,/\,$0.033$ & $0.023$\,/\,$0.032$ & $0.022$\,/\,$0.032$ & $0.023$\,/\,$0.031$ & $0.024$\,/\,$0.031$ \\
$10$  & $0.039$\,/\,$0.045$ & $0.039$\,/\,$0.045$ & $0.039$\,/\,$0.045$ & $0.037$\,/\,$0.045$ & $0.038$\,/\,$0.045$ & $0.038$\,/\,$0.044$ \\
$20$  & $0.058$\,/\,$0.064$ & $0.058$\,/\,$0.064$ & $0.058$\,/\,$0.064$ & $0.058$\,/\,$0.064$ & $0.058$\,/\,$0.063$ & $0.057$\,/\,$0.063$ \\
$50$  & $0.078$\,/\,$0.082$ & $0.087$\,/\,$0.092$ & $\mathbf{0.097}$\,/\,$\mathbf{0.101}$ & $0.097$\,/\,$0.101$ & $0.098$\,/\,$0.101$ & $0.098$\,/\,$0.100$ \\
$100$ & $0.078$\,/\,$0.082$ & $0.089$\,/\,$0.093$ & $0.106$\,/\,$0.113$ & $0.127$\,/\,$0.142$ & $0.140$\,/\,$0.162$ & $0.141$\,/\,$0.156$ \\
\bottomrule
\end{tabular}}
\end{minipage}
\end{center}

\paragraph{Permutation test and median.}
\Cref{tab:app_sens_perm} shows that $n_\text{perm}$, $\alpha$, and $n_\text{RDC}$ barely affect $\learnability$.
On the track force data, the mean absolute deviation of $\learnability$ from its value for $n_\text{perm} = 4999$ is at most $4.0 \cdot 10^{-4}$ in every cell of the $n_\text{perm}$--$\alpha$ grid, less than a fifth of the spread over seeds, and at most $1.1 \cdot 10^{-4}$ for $n_\text{perm} \geq 499$, an order of magnitude below it.
On the random disturbance data, the mean p-value is between $0.40$ and $0.43$ for every $n_\text{perm}$.

\paragraph{Tolerance.}
The tolerance $\tau$ has no effect over many orders of magnitude around its default.
The set of used components, and thus every result, is identical for $\tau \in [0, 0.1]$ on the track force data and for $\tau \in [0, 3 \cdot 10^{-3}]$ on the random disturbance data.
The study shows that the default lies at least five orders of magnitude below the smallest relative range of a genuine component ($3 \cdot 10^{-3}$ on the random disturbance data).

In summary, the test decision is insensitive to all hyperparameters in the grid on these two data sets, and across these hyperparameters $\learnability$ of the high-learnability case varies only between $0.958$ and $0.983$, or between $0.976$ and $0.983$ for $k_\text{RDC} \geq 10$. 
For every setting of $n_\text{perm}$, $\alpha$, $n_\text{RDC}$, $k_\text{RDC}$, $s_\text{RDC}$, and $\tau$ for which at least one residual component remains, the test rejects independence on the track force data and does not reject it on the random disturbance data in all five seeds.
Hence, no hyperparameter setting in the grid changes the outcome of the test on these two data sets.
The default values lie in the flat region and thus mainly improve repeatability.

\begin{center}
\begin{minipage}{\linewidth}
\centering
\captionof{table}{Sensitivity of the learnability test to the number of re-pairings $n_\text{perm}$, the significance level $\alpha$, and the number of RDCs $n_\text{RDC}$ in the median, varied one at a time with the others at their defaults $n_\text{perm} = 999$, $\alpha = 0.05$, and $n_\text{RDC} = 5$. rej.\ is the number of seeds in which it rejects independence on the random disturbance data, with mean p-value $p$.}
\label{tab:app_sens_perm}
{\footnotesize
\setlength{\tabcolsep}{4pt}
\begin{tabular}{lcccccc}
\toprule
 & \multicolumn{2}{c}{\textbf{Track force}} & \multicolumn{4}{c}{\textbf{Random disturbances}} \\
\cmidrule(lr){2-3} \cmidrule(lr){4-7}
\textbf{Value} & $\learnability$ & $q_\alpha$ & $\rdc$ & $q_\alpha$ & $p$ & rej. \\
\midrule
\multicolumn{7}{l}{\textit{Number of re-pairings $n_\text{perm}$}} \\
$49$   & $0.9794\pm0.0023$ & $0.372$ & $0.097$ & $0.102$ & $0.432$ & 0\,/\,5 \\
$99$   & $0.9794\pm0.0022$ & $0.371$ & $0.097$ & $0.101$ & $0.418$ & 0\,/\,5 \\
$199$  & $0.9795\pm0.0022$ & $0.369$ & $0.097$ & $0.101$ & $0.412$ & 0\,/\,5 \\
$499$  & $0.9794\pm0.0022$ & $0.371$ & $0.097$ & $0.101$ & $0.401$ & 0\,/\,5 \\
$999$  & $0.9794\pm0.0022$ & $0.371$ & $0.097$ & $0.101$ & $0.399$ & 0\,/\,5 \\
$1999$ & $0.9794\pm0.0022$ & $0.371$ & $0.097$ & $0.101$ & $0.399$ & 0\,/\,5 \\
\midrule
\multicolumn{7}{l}{\textit{Significance level $\alpha$}} \\
$0.01$  & $0.9785\pm0.0024$ & $0.396$ & $0.097$ & $0.104$ & $0.399$ & 0\,/\,5 \\
$0.025$ & $0.9791\pm0.0023$ & $0.382$ & $0.097$ & $0.102$ & $0.399$ & 0\,/\,5 \\
$0.05$  & $0.9794\pm0.0022$ & $0.371$ & $0.097$ & $0.101$ & $0.399$ & 0\,/\,5 \\
$0.1$   & $0.9799\pm0.0021$ & $0.357$ & $0.097$ & $0.100$ & $0.399$ & 0\,/\,5 \\
\midrule
\multicolumn{7}{l}{\textit{Number of RDCs in the median $n_\text{RDC}$}} \\
$1$ & $0.9798\pm0.0032$ & $0.403$ & $0.097$ & $0.103$ & $0.445$ & 0\,/\,5 \\
$3$ & $0.9801\pm0.0032$ & $0.380$ & $0.097$ & $0.102$ & $0.516$ & 0\,/\,5 \\
$5$ & $0.9794\pm0.0022$ & $0.371$ & $0.097$ & $0.101$ & $0.399$ & 0\,/\,5 \\
$7$ & $0.9798\pm0.0023$ & $0.365$ & $0.098$ & $0.101$ & $0.354$ & 0\,/\,5 \\
$9$ & $0.9802\pm0.0021$ & $0.360$ & $0.098$ & $0.101$ & $0.367$ & 0\,/\,5 \\
\bottomrule
\end{tabular}}
\end{minipage}
\end{center}

\section{Evaluation Details}\label{app:evaluation_details}

All MPCs are implemented in the do-mpc Python library~\citep{fiedler_DompcFAIRNonlinear_2023} using CasADi~\citep{andersson_CasADiSoftwareFramework_2019} with IPOPT~\citep{wachter_ImplementationInteriorpointFilter_2006} as the solver.
The RL agents are trained with Stable-Baselines3~\citep{raffin_StableBaselines3ReliableReinforcement_2021} on PyTorch~\citep{ansel_PyTorch2Faster_2024}.
The environments are custom implementations in the Gymnasium library~\citep{towers_GymnasiumStandardInterface_2025}.
All experiments were run on a single Ubuntu workstation equipped with two NVIDIA A100-PCIE-40GB GPUs and dual Intel Xeon Platinum 8380 CPUs. \cref{tab:app_hardware} lists the hardware and operating system, and \cref{tab:app_software} the versions of the main software libraries. Any library not listed is used at the default version pulled by its respective dependency.
For the knowledge advantage test we use $n_b =5$ and $\lambda=0.05$ throughout all experiments.
For the learnability test, we use $k_\text{RDC} = 50$, $s_\text{RDC} = 1/6$, $n_\text{RDC} = 5$, $n_\text{perm} = 999$, $\alpha = 0.05$, and $\tau = \sqrt{\epsilon_\text{mach}} \approx 1.5 \cdot 10^{-8}$.

\begin{center}
\begin{minipage}[t]{0.58\linewidth}
\centering
\captionof{table}{Hardware and operating system used for all experiments.}
\label{tab:app_hardware}
{\small
\begin{tabular}{@{}l>{\raggedright\arraybackslash}p{5.4cm}@{}}
\toprule
\textbf{Component} & \textbf{Specification} \\
\midrule
CPU & 2$\times$ Intel Xeon Platinum 8380 @ \qty{2.30}{\giga\hertz} (40 cores / 80 threads each) \\
Memory & 2.0\,TiB RAM \\
GPU & 2$\times$ NVIDIA A100-PCIE-40GB (40\,GB HBM2 each) \\
Driver / CUDA & 595.71.05 / CUDA 13.2 \\
OS & Ubuntu 24.04.4 LTS (kernel 6.8.0-124-generic) \\
\bottomrule
\end{tabular}}
\end{minipage}\hfill
\begin{minipage}[t]{0.38\linewidth}
\centering
\captionof{table}{Versions of the main software libraries.}
\label{tab:app_software}
{\small
\begin{tabular}{ll}
\toprule
\textbf{Library} & \textbf{Version} \\
\midrule
PyTorch & 2.10.0+cu126 \\
Stable-Baselines3 & 2.7.1 \\
Gymnasium & 1.2.3 \\
do-mpc & 5.1.1 \\
CasADi (with IPOPT) & 3.7.2 \\
Optuna & 4.8.0 \\
NumPy & 2.2.6 \\
\bottomrule
\end{tabular}}
\end{minipage}
\end{center}

In the following subsections, we describe the governing dynamics and physical parameters of the three benchmark systems and of the spring-mass-damper chain used in the ablation study.

\subsection{Cart Pole}\label{app:plant_cartpole}

The cart pole consists of a pole balanced on a cart that moves along one horizontal axis~\citep{barto_NeuronlikeAdaptiveElements_1983}.
Its state is $x = [s, v, \theta, \omega]^\top$, where $s$ and $v$ are the position and velocity of the cart and $\theta$ and $\omega$ are the angle and angular velocity of the pole, with $\theta = 0$ when the pole is upright.
The control input is the horizontal force $F$ on the cart.
The pole is modeled as a point mass $m_h$ at the end of a massless rod of length $l$, giving the governing dynamics
\begin{align*}
    \dot{v} &= \frac{F - m_h l \omega^2 \sin\theta + m_h g \sin\theta \cos\theta}{m_c + m_h \sin^2\theta}, \\
    \dot{\omega} &= \frac{F \cos\theta + m_h g \sin\theta \cos^2\theta}{l \left( m_c + m_h \sin^2\theta \right)} + \frac{g \sin\theta}{l} \\
    &\quad - \frac{m_h \omega^2 \sin\theta \cos\theta}{m_c + m_h \sin^2\theta},
\end{align*}
with $\dot{s} = v$ and $\dot{\theta} = \omega$, where $m_c$ is the cart mass, $m_h$ the pole head mass, $l$ the pole length, and $g$ the gravitational acceleration.

In the \textit{track force} scenario, the true plant additionally experiences a horizontal track force $F_\text{tr}$ on the cart that the model neglects.
It combines Coulomb friction between cart and track, smoothed around zero velocity, and the downhill force of an uneven track with height profile $h(s) = A \sin(2 \pi s / \lambda_\text{tr})$,
\begin{align*}
    F_\text{tr}(s, v) = -\mu_c (m_c + m_h) g \tanh\!\left( \frac{v}{v_s} \right) - (m_c + m_h) g \, h'(s) ,
\end{align*}
where the second term is a small-slope approximation (maximum slope $2 \pi A / \lambda_\text{tr} \approx \ang{7.2}$).
The force enters the dynamics like the input $F$, i.e., it causes the accelerations
\begin{align*}
    \ddot{s}_\text{tr} = \frac{F_\text{tr}}{m_c + m_h \sin^2\theta} , \quad \ddot{\theta}_\text{tr} = \frac{\cos\theta}{l} \, \ddot{s}_\text{tr} .
\end{align*}
These are evaluated at the state at the beginning of each sampling interval and added to the velocities after the nominal integration step, i.e., $\vectorDisturbances = [0, \Delta t \, \ddot{s}_\text{tr}, 0, \Delta t \, \ddot{\theta}_\text{tr}]^\top$ with the sample time $\Delta t$.
The resulting disturbance depends only on the state, not on the input or on time, and the residuals of $s$ and $\theta$ are zero.
With $\Delta t \, \mu_c (m_c + m_h) g / (m_c v_s) = 0.76 < 1$, the friction cannot reverse the cart velocity within one step.

For the evaluation, a total of $n_\indexMeasured = 50$ trajectories is generated, initial states are sampled from $\setStates(0) = [-0.3, 0.3] \times [-0.2,0.2] \times [-0.2,0.2] \times [-0.2,0.2]$.
The physical parameters are listed in \cref{tab:app_plant_cartpole}.

\begin{center}
\begin{minipage}{\linewidth}
\centering
\captionof{table}{Cart pole parameters.}
\label{tab:app_plant_cartpole}
{\small
\begin{tabular}{clc}
\toprule
\textbf{Symbol} & \textbf{Description} & \textbf{Value} \\
\midrule
$m_c$ & Cart mass & \qty{1.0}{\kilo\gram} \\
$m_h$ & Pole head mass & \qty{0.1}{\kilo\gram} \\
$l$ & Pole length & \qty{0.5}{\meter} \\
$g$ & Gravitational acceleration & \qty{9.81}{\meter\per\second\squared} \\
$F$ & Input force range & $[-10, 10]\,\unit{\newton}$ \\
\midrule
\multicolumn{3}{l}{\textit{Track force scenario}} \\
$\mu_c$ & Coulomb friction coefficient & $0.14$ \\
$v_s$ & Friction smoothing velocity & \qty{0.1}{\meter\per\second} \\
$A$ & Track amplitude & \qty{0.02}{\meter} \\
$\lambda_\text{tr}$ & Track wavelength & \qty{1.0}{\meter} \\
\bottomrule
\end{tabular}}
\end{minipage}
\end{center}

\subsection{Continuously Stirred Tank Reactor}\label{app:plant_cstr}

The Continuously Stirred Tank Reactor (CSTR) hosts two simultaneous reactions and has the state $x = [c_A, c_B, T]^\top$, where $c_A$ and $c_B$ are the concentrations of species A and B and $T$ is the reactor temperature.
The control input is the heat input rate $\dot{Q}$.
The governing dynamics are
\begin{align*}
    \dot{c}_A &= \dfrac{F(c_{A,0} - c_A) - k_A c_A V + k_B c_B V}{V}, \\
    \dot{c}_B &= \dfrac{-F c_B + k_A c_A V - k_B c_B V}{V}, \\
    \dot{T} &= \dfrac{F \rho_l C_p (T_0 - T) + \dot{Q} - \Delta H (k_A c_A V - k_B c_B V)}{\rho_l C_p V} ,
\end{align*}
where $F$ is the flow rate, $V$ the reactor volume, $\rho_l$ the liquid density, $C_p$ the heat capacity, $T_0$ the feed temperature, $c_{A,0}$ the feed concentration, and $\Delta H$ the reaction enthalpy.
The reaction rates $k_A$ and $k_B$ follow the Arrhenius formulation
\begin{align*}
    k = k_0 \, e^{\frac{-E_a}{R T}} ,
\end{align*}
where $k_0$ is the pre-exponential factor, $E_a$ the activation energy, and $R$ the gas constant.
The physical parameters are listed in \cref{tab:app_plant_cstr}. 
The objective of the control problem is to maximize the concentration of species B while minimizing energy consumption.
The reactor is operated for a total of $\qty{600}{\minute}$ with a sample time of $\qty{1}{\minute}$.
The mismatched model used by the controller overestimates $k_{0,A}$ by $10\%$ ($55000$) and underestimates $k_{0,B}$ by $5\%$ ($95000$).

\begin{center}
\begin{minipage}{\linewidth}
\centering
\captionof{table}{CSTR parameters.}
\label{tab:app_plant_cstr}
{\small
\begin{tabular}{clc}
\toprule
\textbf{Symbol} & \textbf{Description} & \textbf{Value} \\
\midrule
$\rho_l$ & Density & \qty{780}{\kilo\gram\per\cubic\meter} \\
$c_{A,0}$ & Feed concentration & \qty{0.3}{\kilo\mol\per\cubic\meter} \\
$C_p$ & Heat capacity & \qty{3.25}{\kilo\joule\per\kilo\gram\per\kelvin} \\
$E_{a,A}$ & Activation energy A & \qty{41570}{\kilo\joule\per\kilo\mol} \\
$E_{a,B}$ & Activation energy B & \qty{45727}{\kilo\joule\per\kilo\mol} \\
$F$ & Flow rate & \qty{0.0025}{\cubic\meter\per\minute} \\
$\Delta H$ & Reaction enthalpy & \qty{4157}{\kilo\joule\per\kilo\mol} \\
$k_{0,A}$ & Pre-exp. factor A & \qty{50000}{\per\minute} \\
$k_{0,B}$ & Pre-exp. factor B & \qty{100000}{\per\minute} \\
$R$ & Gas constant & \qty{8.314}{\kilo\joule\per\kilo\mol\per\kelvin} \\
$T_0$ & Feed temperature & \qty{300}{\kelvin} \\
$V$ & Reactor volume & \qty{0.2}{\cubic\meter} \\
$\dot{Q}$ & Input heat range & $[0, 1500]\,\unit{\kilo\joule\per\minute}$ \\
\bottomrule
\end{tabular}}
\end{minipage}
\end{center}

\subsection{2D Quadrotor}\label{app:plant_quad}

The 2D quadrotor is constrained to the vertical $s_x$--$s_z$ plane~\citep{yuan_SafeControlGymUnifiedBenchmark_2022} and has the state $x = [s_x, s_z, v_x, v_z, \theta, \omega]^\top$, where $s_x, s_z$ and $v_x, v_z$ are the position and velocity and $\theta$ and $\omega$ are the pitch angle and angular velocity.
The control input is the pair of front and back rotor thrusts $f_1$ and $f_2$. 
We use the north-east-down (NED) convention, i.e., negative $s_z$ denotes positions above the ground, so that gravity acts in the positive $s_z$ direction.
A total of $30$ trajectories are simulated and used to synthesize a reachset-conformant model.
The governing dynamics are
\begin{align*}
    \dot{v}_x &= \frac{(f_1 + f_2) \sin\theta}{m}, \\
    \dot{v}_z &= g - \frac{(f_1 + f_2) \cos\theta}{m} + w_\text{ge}(s_z), \\
    \dot{\omega} &= \frac{l \, (f_1 - f_2)}{I_{yy}},
\end{align*}
with $\dot{s}_x = v_x$, $\dot{s}_z = v_z$, and $\dot{\theta} = \omega$, where $m$ is the mass, $I_{yy}$ the moment of inertia, $l$ the distance from the center to each rotor, and $g$ the gravitational acceleration.
The term $w_\text{ge}$ captures the ground effect, the additional lift produced when the rotor downwash is deflected by the floor as the quadrotor flies close to the ground.
We model it with the classical rotor-in-ground-effect relation of \citet{cheeseman_EffectGroundHelicopter_1955}, which expresses the in-ground-effect thrust $T_\text{IGE}$ of a rotor relative to its out-of-ground-effect value $T_\text{OGE}$ as
\begin{align*}
    \frac{T_\text{IGE}}{T_\text{OGE}} = \frac{1}{1 - \chi}, \quad \chi = \left( \frac{r}{4 h} \right)^{2}, \quad h = \max(-s_z, \underline{h}),
\end{align*}
where $r$ is the rotor radius and $h$ the altitude above the ground.
Evaluating the surplus thrust at hover, where $T_\text{OGE} = m g$, turns it into an additive disturbance on the vertical acceleration,
\begin{align*}
    w_\text{ge}(s_z) = -\min\!\left( g \, \frac{\chi}{1 - \chi}, \; \alpha_\text{ge} \, g \right) .
\end{align*}
The negative sign follows from the NED convention. 
The surplus thrust acts upward and therefore reduces $\dot{v}_z$.
The relation is derived for $h / r \gtrsim 0.5$ and diverges as $h \to 0$, so the altitude is clamped from below at $\underline{h}$. 
The cap $\alpha_\text{ge}$ is a numerical guard only, as the clamp binds first ($w_\text{ge}(-\underline{h}) \approx -0.28\,g$).
The benchmark operates at altitudes of \qtyrange{2}{7}{\centi\meter}, over which the effect spans $0.65\%$ to $8.6\%$ of $g$, i.e., it is small in magnitude but strongly altitude-dependent.
The ground effect is not part of the model available to the controller and thus constitutes an unmodeled disturbance.
The physical parameters are listed in \cref{tab:app_plant_quad}.

\begin{center}
\begin{minipage}{\linewidth}
\centering
\captionof{table}{2D quadrotor parameters. Mass, arm length, and rotor radius correspond to a Crazyflie 2.x class micro quadrotor. The ground effect cap $\alpha_\text{ge}$ is given as a fraction of $g$.}
\label{tab:app_plant_quad}
{\small
\begin{tabular}{clc}
\toprule
\textbf{Symbol} & \textbf{Description} & \textbf{Value} \\
\midrule
$m$ & Mass & \qty{0.027}{\kilo\gram} \\
$I_{yy}$ & Moment of inertia & \qty{1.4e-5}{\kilo\gram\meter\squared} \\
$l$ & Arm length & \qty{0.0397}{\meter} \\
$g$ & Gravitational acceleration & \qty{9.81}{\meter\per\second\squared} \\
$f_i$ & Per-rotor thrust range & $[0, 0.4]\,\unit{\newton}$ \\
$r$ & Rotor radius & \qty{0.0225}{\meter} \\
$\underline{h}$ & Altitude lower bound & \qty{0.012}{\meter} \\
$\alpha_\text{ge}$ & Ground effect cap & $0.30$ \\
\bottomrule
\end{tabular}}
\end{minipage}
\end{center}

\subsection{Spring-Mass-Damper Chain}\label{app:plant_chain}

The ablation study in \cref{sec:ablation} uses a chain of $N$ masses on a line, where mass $1$ is attached to a wall and mass $N$ is free.
Neighboring masses, as well as the wall and mass $1$, are connected by a hardening spring with force $F_s(d) = k_1 d + k_3 d^3$ for the elongation $d$ and by a linear damper with coefficient $c$.
The state is $x = [s_1, \ldots, s_N, v_1, \ldots, v_N]^\top$, where $s_i$ and $v_i$ are the position and velocity of mass $i$.
The state is fully measured, so that $n_\vectorOutputs = 2N$.
The control input is the force $F$ on mass $1$.
The governing dynamics are
\begin{align*}
    \dot{s}_i &= v_i, \\
    m \dot{v}_i &= - F_s(s_i - s_{i-1}) - c \, (v_i - v_{i-1}) + F_s(s_{i+1} - s_i) + c \, (v_{i+1} - v_i) + \delta_{i1} F,
\end{align*}
for $i = 1, \ldots, N$, where $m$ is the mass of each body, $\delta_{i1}$ is the Kronecker delta, $s_0 = v_0 = 0$ represents the wall, and the terms of the right neighbor are omitted for $i = N$.

The model available to the test uses the nominal parameters in \cref{tab:app_plant_chain}.
The true plant has a $50\%$ stiffer cubic spring term, $k_3 = \qty{3}{\newton\per\cubic\meter}$ instead of $\qty{2}{\newton\per\cubic\meter}$.
In addition, an i.i.d.\ disturbance, uniformly sampled from $[-10^{-3}, 10^{-3}]\,\unit{\meter\per\second}$, is added to each velocity after every integration step.
The resulting residual of the nominal model is state-dependent.

For each $N$, a total of $n_\indexMeasured = 50$ trajectories of $\qty{10}{\second}$ with sample time $\Delta t = \qty{0.05}{\second}$ is generated, with initial states sampled uniformly from $\setStates(0) = [-0.3, 0.3]^{2N}$.
The trajectories are recorded in closed loop with the PD controller $F = -K_p s_1 - K_d v_1 + F_d$ on mass $1$, where the dither $F_d$ is sampled uniformly from $[-0.5, 0.5]\,\unit{\newton}$ in every step.
The PD controller only serves to excite the system during data collection.
For Part~I, the nominal and the oracle MPC regulate all masses to rest at the origin.

\begin{center}
\begin{minipage}{\linewidth}
\centering
\captionof{table}{Spring-mass-damper chain parameters. The model and the true plant differ only in $k_3$.}
\label{tab:app_plant_chain}
{\small
\begin{tabular}{clc}
\toprule
\textbf{Symbol} & \textbf{Description} & \textbf{Value} \\
\midrule
$m$ & Mass of each body & \qty{1.0}{\kilo\gram} \\
$k_1$ & Linear spring stiffness & \qty{1.0}{\newton\per\meter} \\
$k_3$ & Cubic spring stiffness (model / plant) & $2.0 / 3.0\,\unit{\newton\per\cubic\meter}$ \\
$c$ & Damping coefficient & \qty{0.2}{\newton\second\per\meter} \\
\midrule
\multicolumn{3}{l}{\textit{Data generation}} \\
$K_p$ & PD position gain & \qty{2.0}{\newton\per\meter} \\
$K_d$ & PD velocity gain & \qty{1.0}{\newton\second\per\meter} \\
$F_d$ & Input dither range & $[-0.5, 0.5]\,\unit{\newton}$ \\
\bottomrule
\end{tabular}}
\end{minipage}
\end{center}

\section{Reachset-Conformant System Identification}\label{app:test_identification}

The sets $\setStates(0)$, $\setInputs$, and $\setDisturbances$ of the main text are obtained with the reachset-conformance method of~\citet{lutzow_ScalableReachsetConformantIdentification_2024}, applied to a discrete-time nonlinear model of the plant obtained by fourth-order Runge-Kutta discretization.
The disturbance enters the identification as an additional additive input. 
The augmented input is $[\vectorInputs^\top, \vectorDisturbances^\top]^\top$, with the generator of the true control input fixed to zero so that all identified uncertainty is attributed to $\setDisturbances$.
The identification keeps the centers of $\setInputs$ and $\setDisturbances$ fixed and only scales their generators, so the center $\bar{\vectorDisturbances}$ of $\setDisturbances$ is chosen beforehand.
For an additive disturbance, $f(\vectorStates, \vectorInputs, \vectorDisturbances) = f(\vectorStates, \vectorInputs, 0) + \vectorDisturbances$, and a fully measured state, $g(\vectorStates, \vectorInputs, \vectorDisturbances) = \vectorStates$, every transition of the identification data yields the one-step residual
\begin{equation*}
    \vectorResiduals^{(\indexMeasured)}(\indexTime+1) = \vectorStates^{(\indexMeasured)}(\indexTime+1) - f\big(\vectorStates^{(\indexMeasured)}(\indexTime), \vectorInputs^{(\indexMeasured)}(\indexTime), 0\big) ,
\end{equation*}
i.e., the residual~\cref{eq:residual_y} with $\hat{\vectorStates}(\indexTime) = \vectorStates(\indexTime)$.
We place the center at the per-dimension midrange
\begin{equation*}
    \bar{\vectorDisturbances}_i = \frac{1}{2} \Big( \max_{\indexMeasured, \indexTime} \vectorResiduals_i^{(\indexMeasured)}(\indexTime) + \min_{\indexMeasured, \indexTime} \vectorResiduals_i^{(\indexMeasured)}(\indexTime) \Big) , \quad i = 1, \ldots, n_\vectorStates ,
\end{equation*}
which minimizes $\max_{\indexMeasured, \indexTime} |\vectorResiduals_i^{(\indexMeasured)}(\indexTime) - \bar{\vectorDisturbances}_i|$, i.e., it is the center of the smallest box containing all one-step residuals.
Since $\setDisturbances$ must also contain the multistep reachable sets, the identified set is not guaranteed to be the smallest one for this center.
If the state is not fully measured, the residuals cannot be formed and we use $\bar{\vectorDisturbances} = 0$.
The center of the control input is zero.
For all three benchmarks the identified initial-state uncertainty is the singleton, i.e., the measured initial states are matched exactly, and the entire model mismatch is captured by $\setDisturbances$.

The resulting sets are listed in \cref{tab:app_sets_cartpole} for the cart pole and in \cref{tab:app_sets_cstr_quad} for the CSTR and the quadrotor, together with the initial-state sets $\setStates(0)$ the test samples from.
Each identified set is given by its center $\bar{\vectorDisturbances}_i$ and half-width $a_i$, i.e., the interval $[\bar{\vectorDisturbances}_i - a_i, \bar{\vectorDisturbances}_i + a_i]$; values below $10^{-15}$ in magnitude are numerical zeros and are reported as $0$.
For the cart pole, the centers in the \textit{velocity noise} and \textit{state noise} scenarios are negligible compared with the half-widths, consistent with the zero-mean noise injected there.
In the \textit{velocity noise} scenario, only the velocity dimension is excited, and in the \textit{state noise} scenario, all four dimensions are.
In the \textit{track force} scenario, the state-dependent track force appears mainly on $v$ and $\omega$.
For the quadrotor, the ground effect is recovered almost exclusively on the vertical velocity dimension, consistent with its derivation as an additive disturbance on $\dot{v}_z$.
The identified interval $[-3.60{\cdot}10^{-2}, -3.37{\cdot}10^{-3}]$ does not contain zero, i.e., the ground effect acts as a one-sided bias.

\begin{center}
\begin{minipage}{\linewidth}
\centering
\captionof{table}{Initial-state set $\setStates(0)$ and identified disturbance sets $\setDisturbances$ of the cart pole for the three disturbance scenarios, given by center $\bar{\vectorDisturbances}_i$ and half-width $a_i$.}
\label{tab:app_sets_cartpole}
{\footnotesize
\setlength{\tabcolsep}{3pt}
\begin{tabular}{lccccccc}
\toprule
 & & \multicolumn{2}{c}{\textbf{Velocity noise}} & \multicolumn{2}{c}{\textbf{State noise}} & \multicolumn{2}{c}{\textbf{Track force}} \\
\cmidrule(lr){3-4} \cmidrule(lr){5-6} \cmidrule(lr){7-8}
 & $\setStates(0)$ & $\bar{\vectorDisturbances}_i$ & $a_i$ & $\bar{\vectorDisturbances}_i$ & $a_i$ & $\bar{\vectorDisturbances}_i$ & $a_i$ \\
\midrule
$s$ & $[-0.3, 0.3]$ & $0$ & $1.66{\cdot}10^{-5}$ & $1.65{\cdot}10^{-7}$ & $9.28{\cdot}10^{-3}$ & $0$ & $3.13{\cdot}10^{-2}$ \\
$v$ & $[-0.2, 0.2]$ & $-7.77{\cdot}10^{-7}$ & $1.03{\cdot}10^{-2}$ & $-1.55{\cdot}10^{-7}$ & $2.39{\cdot}10^{-2}$ & $-5.67{\cdot}10^{-3}$ & $1.57{\cdot}10^{-1}$ \\
$\theta$ & $[-0.2, 0.2]$ & $0$ & $0$ & $1.99{\cdot}10^{-8}$ & $2.09{\cdot}10^{-3}$ & $0$ & $0$ \\
$\omega$ & $[-0.2, 0.2]$ & $0$ & $0$ & $-7.48{\cdot}10^{-9}$ & $2.48{\cdot}10^{-1}$ & $-1.14{\cdot}10^{-2}$ & $1.84$ \\
\bottomrule
\end{tabular}}
\end{minipage}
\end{center}

\begin{center}
\begin{minipage}{\linewidth}
\centering
\captionof{table}{Initial-state sets $\setStates(0)$ and identified disturbance sets $\setDisturbances$ of the CSTR and the quadrotor, given by center $\bar{\vectorDisturbances}_i$ and half-width $a_i$.}
\label{tab:app_sets_cstr_quad}
{\small
\begin{tabular}{lccc}
\toprule
 & $\setStates(0)$ & $\bar{\vectorDisturbances}_i$ & $a_i$ \\
\midrule
\multicolumn{4}{l}{\textit{CSTR}} \\
$c_A$ & $[0.2, 0.4]$ & $3.36{\cdot}10^{-4}$ & $1.37{\cdot}10^{-3}$ \\
$c_B$ & $[0.0, 0.1]$ & $-3.31{\cdot}10^{-4}$ & $1.40{\cdot}10^{-3}$ \\
$T$ & $[290, 310]$ & $5.71{\cdot}10^{-4}$ & $6.10{\cdot}10^{-1}$ \\
\midrule
\multicolumn{4}{l}{\textit{Quadrotor}} \\
$s_x$ & $[-0.10, 0.10]$ & $0$ & $0$ \\
$s_z$ & $[-0.07, -0.02]$ & $0$ & $1.06{\cdot}10^{-9}$ \\
$v_x$ & $[-0.05, 0.05]$ & $0$ & $0$ \\
$v_z$ & $[-0.05, 0.05]$ & $-1.97{\cdot}10^{-2}$ & $1.63{\cdot}10^{-2}$ \\
$\theta$ & $[-0.05, 0.05]$ & $0$ & $0$ \\
$\omega$ & $[-0.10, 0.10]$ & $0$ & $0$ \\
\bottomrule
\end{tabular}}
\end{minipage}
\end{center}

\section{MPC Configuration}\label{app:mpc}

\cref{tab:app_mpc} lists the weights of the quadratic stage cost, the prediction horizon $n_p$, and the reference for the MPCs used in the three benchmarks.
$Q$ weights the state deviation from the reference, $R$ the control input, and $R_\Delta$ the input rate of change.
All three configurations are the outcome of the nominal MPC tuning studies described in \cref{app:optuna_mpc}.
For tuning, the evaluation cost functions described in \cref{tab:app_eval_cost} are used.
These are also used to evaluate and compare the MPC and the RL policy.

\begin{center}
\begin{minipage}{\linewidth}
\centering
\captionof{table}{Evaluation cost $J$ used both as the RL reward weights and for MPC tuning. For the quadrotor, $R_\text{eval}$ applies per rotor.}
\label{tab:app_eval_cost}
{\small
\begin{tabular}{lccc}
\toprule
 & \textbf{Cart pole} & \textbf{CSTR} & \textbf{Quadrotor} \\
\midrule
$Q_\text{eval}$ & $\text{diag}(20, 40, 200, 5)$ & $\text{diag}(10, 100, 0.01)$ & $\text{diag}(100, 100, 10, 10, 1, 1)$ \\
$R_\text{eval}$ & $0.01$ & $10^{-4}$ & $0.01$ \\
\bottomrule
\end{tabular}}
\end{minipage}
\end{center}

\begin{center}
\begin{minipage}{\linewidth}
\centering
\captionof{table}{Configuration of the MPCs used for our proposed screening test based on hyperparameter tuning. For the quadrotor, $R$ and $R_\Delta$ apply per rotor.}
\label{tab:app_mpc}
{\footnotesize
\setlength{\tabcolsep}{3pt}
\begin{tabular}{lccc}
\toprule
 & \textbf{Cart pole} & \textbf{CSTR} & \textbf{Quadrotor} \\
\midrule
Sample time & \qty{0.05}{\second} & \qty{1.0}{\minute} & \qty{0.05}{\second} \\
Horizon $n_p$ & 63 & 73 & 20 \\
$Q$ & $\text{diag}(52.57, 105.97, 200, 3.32)$ & $\text{diag}(4.95, 100, 0.154)$ & $\text{diag}(100, 71.79, 7.74, 11.14, 6.48, 1.02)$ \\
$R$ & $0.0861$ & $5.58 {\cdot} 10^{-4}$ & $0.833$ \\
$R_\Delta$ & $1.08 {\cdot} 10^{-3}$ & $5.04 {\cdot} 10^{-4}$ & $8.22 {\cdot} 10^{-4}$ \\
Reference & $(0, 0, 0, 0)^\top$ & $(0.15, 0.2, 310.0)^\top$ & $(0, -0.03, 0, 0, 0, 0)^\top$ \\
\bottomrule
\end{tabular}}
\end{minipage}
\end{center}

\subsection{MPC Tuning Space}\label{app:optuna_mpc}

The MPC of every benchmark is tuned with its own Optuna study of $40$ trials on the nominal plant. 
The objective is the mean closed-loop cost over $6$ initial states with $2$ repeats each, evaluated with the evaluation cost function reported in \cref{tab:app_eval_cost}. 
The MPC solution is invariant to a common scaling of $(Q, R, R_\Delta)$.
Searching all weights would therefore spend one dimension of the budget on a null direction, so one entry of $Q$ is pinned per benchmark to fix the overall scale, and the study effectively searches the ratios relative to it.
The search spaces are listed in \cref{tab:opt_mpc}.

\begin{center}
\begin{minipage}{\linewidth}
\centering
\captionof{table}{MPC tuning search spaces, one block per benchmark. All continuous parameters are sampled log-uniformly.}
\label{tab:opt_mpc}
{\small
\begin{tabular}{ll}
\toprule
\textbf{Hyperparameter} & \textbf{Search space} \\
\midrule
\multicolumn{2}{l}{\textit{Cart pole} (pinned: $Q_\theta = 200$)} \\
$Q_s$, $Q_v$, $Q_\omega$ & $[1, 80]$, $[1, 120]$, $[0.5, 60]$ \\
$R$, $R_\Delta$ & $[5 \cdot 10^{-4}, 0.2]$, $[10^{-3}, 0.3]$ \\
Horizon $n_p$ & integer $[20, 90]$ \\
\midrule
\multicolumn{2}{l}{\textit{CSTR} (pinned: $Q_{c_B} = 100$)} \\
$Q_{c_A}$, $Q_T$ & $[1, 100]$, $[10^{-3}, 1]$ \\
$R$, $R_\Delta$ & $[10^{-6}, 10^{-2}]$, $[10^{-5}, 10^{-1}]$ \\
Horizon $n_p$ & integer $[15, 80]$ \\
\midrule
\multicolumn{2}{l}{\textit{Quadrotor} (pinned: $Q_{s_x} = 100$)} \\
$Q_{s_z}$, $Q_{v_x}$, $Q_{v_z}$ & $[10, 500]$, $[1, 100]$, $[1, 100]$ \\
$Q_\theta$, $Q_\omega$ & $[0.1, 20]$, $[0.1, 20]$ \\
$R$, $R_\Delta$ & $[10^{-3}, 1]$, $[10^{-5}, 10^{-2}]$ \\
Horizon $n_p$ & integer $[10, 60]$ \\
\bottomrule
\end{tabular}}
\end{minipage}
\end{center}

\section{RL Tuning and Training Configuration}\label{app:rl}

Both SAC and PPO are trained for every benchmark.
The hyperparameters used are the result of hyperparameter tuning using the Optuna library~\citep{akiba_OptunaNextgenerationHyperparameter_2019}.
GAE denotes generalized advantage estimation.
\Cref{tab:app_rl_cp_sac,tab:app_rl_cp_ppo} list the cart pole hyperparameters, for the three scenarios of the main text: \textit{velocity noise}, with uniform noise on $v$ sampled from $[-0.01, 0.01]$; \textit{state noise}, with uniform disturbances on all four states sampled from $[-0.002, 0.002]^4$; and \textit{track force}, with the state-dependent track force of \cref{app:plant_cartpole}. \Cref{tab:app_rl_cstr_sac,tab:app_rl_cstr_ppo} list the CSTR hyperparameters, and \cref{tab:app_rl_quad_sac,tab:app_rl_quad_ppo} the quadrotor hyperparameters.
All agents are trained with Stable-Baselines3 using the default values of the respective algorithm for every hyperparameter not listed, and every configuration is trained across five seeds.
The reward function of the first two benchmarks is simply the negative quadratic cost using the untuned evaluation weights $Q_\text{eval}$ and $R_\text{eval}$ of \cref{tab:app_eval_cost}.
The reward function of the quadrotor benchmark is described in \cref{app:quad_env}.

\subsection{Cart Pole}

\begin{center}
\begin{minipage}{\linewidth}
\centering
\captionof{table}{Cart pole SAC hyperparameters for the three disturbance scenarios of the main text.}
\label{tab:app_rl_cp_sac}
{\footnotesize
\setlength{\tabcolsep}{4pt}
\begin{tabular}{lccc}
\toprule
 & \textbf{Velocity noise} & \textbf{State noise} & \textbf{Track force} \\
\midrule
Learning rate & $2.46{\cdot}10^{-4}$ & $1.73{\cdot}10^{-4}$ & $4.01{\cdot}10^{-4}$ \\
Batch size & 128 & 256 & 128 \\
Discount factor $\gamma$ & 0.999 & 0.999 & 0.99 \\
Network architecture & $[512, 512]$ & $[256, 256]$ & $[256, 256]$ \\
Buffer size & 10{,}000 & 30{,}000 & 60{,}000 \\
Target smoothing $\tau$ & 0.01 & 0.02 & 0.01 \\
Train freq. / grad. steps & 8 / 4 & 4 / 1 & 1 / 1 \\
\midrule
\multicolumn{4}{p{8cm}}{\textit{Shared:} total time steps $= 300{,}000$, episode length $= 200$ steps, time step $= \qty{0.05}{\second}$, learning starts $= 5{,}000$, automatic entropy coefficient, linear learning rate decay.} \\
\bottomrule
\end{tabular}}
\end{minipage}
\end{center}

\begin{center}
\begin{minipage}{\linewidth}
\centering
\captionof{table}{Cart pole PPO hyperparameters for the three disturbance scenarios of the main text.}
\label{tab:app_rl_cp_ppo}
{\footnotesize
\setlength{\tabcolsep}{4pt}
\begin{tabular}{lccc}
\toprule
 & \textbf{Velocity noise} & \textbf{State noise} & \textbf{Track force} \\
\midrule
Learning rate & $9.19{\cdot}10^{-4}$ & $1.66{\cdot}10^{-4}$ & $6.03{\cdot}10^{-4}$ \\
Batch size & 64 & 128 & 512 \\
Discount factor $\gamma$ & 0.95 & 0.95 & 0.99 \\
Network architecture & $[512, 512]$ & $[64, 64]$ & $[64, 64]$ \\
Epochs per update & 3 & 20 & 10 \\
GAE $\lambda$ & 0.98 & 0.98 & 0.95 \\
Clip range & 0.4 & 0.3 & 0.2 \\
Entropy coefficient & 0.02 & 0.005 & 0.05 \\
\midrule
\multicolumn{4}{p{8cm}}{\textit{Shared:} total time steps $= 1{,}000{,}000$, episode length $= 200$ steps, time step $= \qty{0.05}{\second}$, rollout length $n_\text{steps} = 2048$, constant learning rate.} \\
\bottomrule
\end{tabular}}
\end{minipage}
\end{center}

\subsection{Continuously Stirred Tank Reactor}

\begin{center}
\begin{minipage}[t]{0.49\linewidth}
\centering
\captionof{table}{CSTR SAC training hyperparameters.}
\label{tab:app_rl_cstr_sac}
{\small
\begin{tabular}{ll}
\toprule
\textbf{Hyperparameter} & \textbf{Value} \\
\midrule
Total time steps & 300{,}000 \\
Episode length & 600 steps \\
Time step & \qty{1.0}{\minute} \\
Learning rate & $5.61 \cdot 10^{-5}$, constant \\
Batch size & 512 \\
Discount factor $\gamma$ & 0.995 \\
Network architecture & $[256, 256]$ (pi/qf) \\
Buffer size & 30{,}000 \\
Learning starts & 10{,}000 \\
Target smoothing $\tau$ & 0.01 \\
Train freq. / gradient steps & 16 / 16 \\
Entropy coefficient & 0.001 \\
\bottomrule
\end{tabular}}
\end{minipage}\hfill
\begin{minipage}[t]{0.49\linewidth}
\centering
\captionof{table}{CSTR PPO training hyperparameters.}
\label{tab:app_rl_cstr_ppo}
{\small
\begin{tabular}{ll}
\toprule
\textbf{Hyperparameter} & \textbf{Value} \\
\midrule
Total time steps & 3{,}000{,}000 \\
Episode length & 600 steps \\
Time step & \qty{1.0}{\minute} \\
Learning rate & $2.33 \cdot 10^{-4}$, constant \\
Batch size & 64 \\
Discount factor $\gamma$ & 0.995 \\
Network architecture & $[256, 256]$ (pi/vf) \\
Rollout length $n_\text{steps}$ & 4096 \\
Epochs per update & 10 \\
GAE $\lambda$ & 0.95 \\
Clip range & 0.3 \\
Entropy coefficient & 0.005 \\
\bottomrule
\end{tabular}}
\end{minipage}
\end{center}

\subsection{2D Quadrotor}\label{app:quad_env}

\begin{center}
\begin{minipage}[t]{0.49\linewidth}
\centering
\captionof{table}{Quadrotor SAC training hyperparameters.}
\label{tab:app_rl_quad_sac}
{\footnotesize
\setlength{\tabcolsep}{4pt}
\begin{tabular}{ll}
\toprule
\textbf{Hyperparameter} & \textbf{Value} \\
\midrule
Total time steps & 250{,}000 \\
Episode length & 100 steps \\
Time step & \qty{0.05}{\second} \\
Learning rate & $2.11 \cdot 10^{-3}$, constant \\
Batch size & 128 \\
Discount factor $\gamma$ & 0.9915 \\
Network architecture & $[256, 256]$ (pi/qf) \\
Buffer size & 75{,}000 \\
Learning starts & 5{,}000 \\
Target smoothing $\tau$ & $1.10 \cdot 10^{-2}$ \\
Train freq. / gradient steps & 1 / 2 \\
Entropy coefficient & auto\_0.1 \\
\bottomrule
\end{tabular}}
\end{minipage}\hfill
\begin{minipage}[t]{0.49\linewidth}
\centering
\captionof{table}{Quadrotor PPO training hyperparameters.}
\label{tab:app_rl_quad_ppo}
{\footnotesize
\setlength{\tabcolsep}{4pt}
\begin{tabular}{ll}
\toprule
\textbf{Hyperparameter} & \textbf{Value} \\
\midrule
Total time steps & 250{,}000 \\
Episode length & 100 steps \\
Time step & \qty{0.05}{\second} \\
Learning rate & $6.76 \cdot 10^{-4}$, constant \\
Batch size & 64 \\
Discount factor $\gamma$ & 0.9705 \\
Network architecture & $[256, 256]$ (pi/vf) \\
Rollout length $n_\text{steps}$ & 2048 \\
Epochs per update & 14 \\
GAE $\lambda$ & 0.9692 \\
Clip range & 0.299 \\
Entropy coefficient & $4.77 \cdot 10^{-2}$ \\
Value function coefficient & 0.761 \\
Max.\ gradient norm & 0.329 \\
Initial policy standard deviation & $e^{-3} \approx 0.05$ \\
\bottomrule
\end{tabular}}
\end{minipage}
\end{center}

The quadrotor environment differs from the other two benchmarks in several regards.
The agent does not command the rotor thrusts directly but a collective deviation $u_\text{c}$ from the hover thrust, bounded to $[-0.08, 0.12]\,\unit{\newton}$, and a normalized differential command $u_\text{d} \in [-1, 1]$ scaled by $0.02\,\unit{\newton}$, which are mapped to the rotor thrusts as $f_{1,2} = f_\text{hover} + u_\text{c} \pm u_\text{d} / 2$.
Decoupling the two channels is necessary because the policy uses a single exploration standard deviation across all action dimensions. 
Applied directly to $f_1$ and $f_2$, an exploration scale that is safe for the collective channel induces angular accelerations that violate the attitude limit within a few steps, since $I_{yy}$ is four orders of magnitude smaller than $m$.
The reward is not the negative quadratic cost used for the other two benchmarks, but a shaped reward built on the per-step cost
\begin{align*}
    c(\indexTime) = \| \vectorStates(\indexTime) - \vectorStates_\text{ref} \|_{Q_\text{eval}}^2 + \| f(\indexTime) - f_\text{hover} \|_{R_\text{eval}}^2 ,
\end{align*}
with the weights of \cref{tab:app_eval_cost}, the rotor thrusts $f = [f_1, f_2]^\top$, the hover thrust $f_\text{hover} = m g / 2$ per rotor, and the norm $\|o\|_M^2=o^\top M o$.
The reward at step $\indexTime$ is then
\begin{align*}
    r(\indexTime) = &\; \exp\!\left( -\frac{c(\indexTime)}{c(0)} \right) + 0.3 \tanh\!\left( \frac{c(\indexTime-1) - c(\indexTime)}{c(0)} \right) \\
    &- 0.2 \ln\!\left( 1 + \exp\!\left( \frac{|\theta(\indexTime)| / \bar{\theta} - 0.7}{0.1} \right) \right) \\
    &- 0.1 \, \frac{\| f(\indexTime) - f_\text{hover} \|_2^2}{\bar{u}^2} - 0.5 \tanh \left( c(\indexTime) \right) \\
    &+ 5 \, \mathbf{1}_\text{goal} - \left( 20 + 3 \left( n_\indexTime - \indexTime \right) \right) \mathbf{1}_\text{unsafe} ,
\end{align*}
with the cost scale $c(0) = 50$, the attitude limit $\bar{\theta} = \qty{0.8}{\radian}$, the action scale $\bar{u} = \qty{0.12}{\newton}$, and the episode length $n_\indexTime = 100$.
The first term rewards proximity to the reference in a bounded manner, the second progresses towards it, the third maintains a margin to the attitude limit, the fourth penalizes deviations from the hover thrust, and the fifth restores a gradient in the near-converged regime, where the first term is already saturated.
The progress term is omitted in the first step of an episode.
The indicator $\mathbf{1}_\text{goal}$ equals one if the episode reaches the horizon $n_\indexTime$ with a state-dependent cost below $c_0$, and $\mathbf{1}_\text{unsafe}$ equals one if $|\theta(\indexTime)| > \bar{\theta}$, which also terminates the episode early. 
This reward structure was found to deliver the most stable training results. 

\subsection{RL Hyperparameter Search Spaces}\label{app:optuna}

This section lists the search spaces of the Optuna studies described in the main text, with one table per benchmark and algorithm (\cref{tab:opt_cp_ppo,tab:opt_cp_sac,tab:opt_cstr_ppo,tab:opt_cstr_sac,tab:opt_quad_sac,tab:opt_quad_ppo}).
Each table gives the range or set searched for every tuned hyperparameter and the parameters held fixed across all trials of the study.
Every RL study comprises $50$ trials with a median pruner, at a per-trial training budget of $100{,}000$ steps for the cart pole, $1{,}000{,}000$ steps for CSTR PPO and $100{,}000$ steps for CSTR SAC, and $250{,}000$ steps for the quadrotor.

\begin{center}
\begin{minipage}[t]{0.49\linewidth}
\centering
\captionof{table}{Cart pole SAC search space.}
\label{tab:opt_cp_sac}
{\footnotesize
\setlength{\tabcolsep}{4pt}
\begin{tabular}{@{}ll@{}}
\toprule
\textbf{Hyperparameter} & \textbf{Search space} \\
\midrule
Learning rate & log-uniform $[10^{-5}, 10^{-3}]$ \\
Batch size & $\{64, 128, 256, 512\}$ \\
Discount factor $\gamma$ & $\{0.95, 0.98, 0.99, 0.995, 0.999\}$ \\
Network architecture & $\{[64, 64], [256, 256], [512, 512]\}$ \\
Target smoothing $\tau$ & $\{0.001, 0.005, 0.01, 0.02\}$ \\
Buffer size & $\{10\text{k}, 30\text{k}, 60\text{k}, 100\text{k}\}$ \\
Train frequency & $\{1, 4, 8\}$ \\
Gradient steps & $\{1, 2, 4\}$ \\
\midrule
\multicolumn{2}{@{}p{\linewidth}@{}}{\textit{Fixed:} learning starts $= 5000$, automatic entropy coefficient, linear learning rate decay.} \\
\bottomrule
\end{tabular}}
\end{minipage}\hfill
\begin{minipage}[t]{0.49\linewidth}
\centering
\captionof{table}{Cart pole PPO search space.}
\label{tab:opt_cp_ppo}
{\footnotesize
\setlength{\tabcolsep}{4pt}
\begin{tabular}{@{}ll@{}}
\toprule
\textbf{Hyperparameter} & \textbf{Search space} \\
\midrule
Learning rate & log-uniform $[10^{-5}, 10^{-3}]$ \\
Batch size & $\{64, 128, 256, 512\}$ \\
Discount factor $\gamma$ & $\{0.95, 0.98, 0.99, 0.995, 0.999\}$ \\
Network architecture & $\{[64, 64], [256, 256], [512, 512]\}$ \\
Epochs per update & $\{3, 5, 10, 20\}$ \\
GAE $\lambda$ & $\{0.9, 0.92, 0.95, 0.98, 1.0\}$ \\
Clip range & $\{0.1, 0.2, 0.3, 0.4\}$ \\
Entropy coefficient & $\{0, 0.005, 0.01, 0.02, 0.05\}$ \\
\midrule
\multicolumn{2}{@{}p{\linewidth}@{}}{\textit{Fixed:} $n_\text{steps} = 2048$, constant learning rate.} \\
\bottomrule
\end{tabular}}
\end{minipage}
\end{center}

\begin{center}
\begin{minipage}[t]{0.49\linewidth}
\centering
\captionof{table}{CSTR SAC search space.}
\label{tab:opt_cstr_sac}
{\footnotesize
\setlength{\tabcolsep}{3pt}
\begin{tabular}{@{}ll@{}}
\toprule
\textbf{Hyperparameter} & \textbf{Search space} \\
\midrule
Learning rate & log-uniform $[10^{-5}, 10^{-3}]$ \\
Batch size & $\{64, 128, 256, 512\}$ \\
Discount factor $\gamma$ & $\{0.95, 0.98, 0.99, 0.995, 0.999\}$ \\
Network architecture & $\{[64, 64], [128, 128], [256, 256]\}$ \\
Target smoothing $\tau$ & $\{0.001, 0.005, 0.01, 0.02, 0.05\}$ \\
Buffer size & $\{10\text{k}, 30\text{k}, 60\text{k}, 100\text{k}\}$ \\
Train freq. (= gradient steps) & $\{1, 4, 8, 16, 32\}$ \\
Entropy coefficient & $\{\text{auto}, 0.1, 0.05, 0.01, 0.001\}$ \\
\midrule
\multicolumn{2}{@{}p{\linewidth}@{}}{\textit{Fixed:} learning starts $= 10{,}000$, constant learning rate.} \\
\bottomrule
\end{tabular}}
\end{minipage}\hfill
\begin{minipage}[t]{0.49\linewidth}
\centering
\captionof{table}{CSTR PPO search space.}
\label{tab:opt_cstr_ppo}
{\footnotesize
\setlength{\tabcolsep}{4pt}
\begin{tabular}{@{}ll@{}}
\toprule
\textbf{Hyperparameter} & \textbf{Search space} \\
\midrule
Learning rate & log-uniform $[10^{-5}, 10^{-3}]$ \\
Batch size & $\{64, 128, 256, 512\}$ \\
Discount factor $\gamma$ & $\{0.95, 0.98, 0.99, 0.995, 0.999\}$ \\
Network architecture & $\{[64, 64], [128, 128], [256, 256]\}$ \\
Epochs per update & $\{3, 5, 10, 20\}$ \\
GAE $\lambda$ & $\{0.9, 0.92, 0.95, 0.98, 1.0\}$ \\
Clip range & $\{0.1, 0.2, 0.3, 0.4\}$ \\
\midrule
\multicolumn{2}{@{}p{\linewidth}@{}}{\textit{Fixed:} $n_\text{steps} = 4096$, entropy coefficient $= 0.005$, constant learning rate.} \\
\bottomrule
\end{tabular}}
\end{minipage}
\end{center}

\begin{center}
\begin{minipage}[t]{0.49\linewidth}
\centering
\captionof{table}{Quadrotor SAC search space.}
\label{tab:opt_quad_sac}
{\footnotesize
\setlength{\tabcolsep}{4pt}
\begin{tabular}{@{}ll@{}}
\toprule
\textbf{Hyperparameter} & \textbf{Search space} \\
\midrule
Learning rate & log-uniform $[10^{-5}, 3 \cdot 10^{-3}]$ \\
Target smoothing $\tau$ & log-uniform $[0.001, 0.02]$ \\
Discount factor $\gamma$ & uniform $[0.97, 0.999]$ \\
Buffer size & $\{25\text{k}, 75\text{k}, 150\text{k}, 250\text{k}\}$ \\
Learning starts & $\{5000, 10000, 20000, 30000\}$ \\
Batch size & $\{128, 256, 512\}$ \\
Train frequency & $\{1, 2, 4, 8\}$ \\
Gradient steps & $\{1, 2, 4, 8\}$ \\
Entropy coefficient & $\{\text{auto}, \text{auto\_0.5}, \text{auto\_0.1}, \text{auto\_0.05}\}$ \\
\midrule
\multicolumn{2}{@{}p{\linewidth}@{}}{\textit{Fixed:} network architecture $[256, 256]$, constant learning rate, seed $42$.} \\
\bottomrule
\end{tabular}}
\end{minipage}\hfill
\begin{minipage}[t]{0.49\linewidth}
\centering
\captionof{table}{Quadrotor PPO search space.}
\label{tab:opt_quad_ppo}
{\footnotesize
\setlength{\tabcolsep}{4pt}
\begin{tabular}{@{}ll@{}}
\toprule
\textbf{Hyperparameter} & \textbf{Search space} \\
\midrule
Learning rate & log-uniform $[10^{-5}, 10^{-3}]$ \\
Rollout length $n_\text{steps}$ & $\{1024, 2048, 4096\}$ \\
Batch size & $\{64, 128, 256\}$ \\
Epochs per update & integer $[5, 15]$ \\
Discount factor $\gamma$ & uniform $[0.97, 0.999]$ \\
GAE $\lambda$ & uniform $[0.90, 0.99]$ \\
Clip range & uniform $[0.10, 0.30]$ \\
Entropy coefficient & log-uniform $[10^{-5}, 5 \cdot 10^{-2}]$ \\
Value function coefficient & uniform $[0.2, 1.0]$ \\
Max. gradient norm & uniform $[0.3, 1.0]$ \\
\midrule
\multicolumn{2}{@{}p{\linewidth}@{}}{\textit{Fixed:} network architecture $[256, 256]$, constant learning rate, initial policy standard deviation $e^{-3}$, seed $42$.} \\
\bottomrule
\end{tabular}}
\end{minipage}
\end{center}

\end{document}